\documentclass[11pt]{article}
\usepackage[margin=1in]{geometry}
\usepackage{amsmath,amssymb,amsthm,mathtools}
\usepackage{booktabs,tabularx}
\usepackage{graphicx}
\usepackage{xcolor}
\usepackage{tikz}
\usetikzlibrary{arrows.meta,positioning,calc}
\usepackage[hidelinks]{hyperref}
\usepackage{microtype}
\newtheorem{theorem}{Theorem}
\newtheorem{proposition}[theorem]{Proposition}
\newtheorem{lemma}[theorem]{Lemma}
\newtheorem{claim}{Claim}
\newtheorem{rem}[theorem]{Remark}

\newtheorem{condition}{Fact}
\DeclareMathOperator{\Tr}{Tr}
\DeclareMathOperator{\ran}{ran}

\newcommand{\reviewedit}[1]{{#1}}
\newcommand{\ip}[2]{\langle #1,#2\rangle_{\rho_\beta}}
\newcommand{\norm}[1]{\left\|#1\right\|}

\title{Spectral Analysis Beyond the Spectral Gap: Rapid Mixing of Quantum Gibbs Samplers in One Dimension}
\author{%
Hongrui Chen\\[0.3em]
{\small Department of Mathematics}\\
{\small Stanford University}\\
{\small Stanford, CA 94305, USA}\\
{\small \href{mailto:hongrui@stanford.edu}{\texttt{hongrui@stanford.edu}}}
\and
Zhiyan Ding\\[0.3em]
{\small Department of Mathematics}\\
{\small University of Michigan}\\
{\small Ann Arbor, MI 48109, USA}\\
{\small \href{mailto:zyding@umich.edu}{\texttt{zyding@umich.edu}}}
}
\date{}
\begin{document}
\maketitle

\begin{abstract}
We prove rapid mixing of a quantum Gibbs sampler for general
one-dimensional non-commuting local spin systems at every fixed positive
temperature. In particular, for a system of $n$ sites, we show that the Kubo--Martin--Schwinger (KMS) detailed-balanced
Gibbs sampler of~\cite{CKG2023} converges from any initial state
to trace-norm error $\epsilon$ in time
$O(\log^4(n)+\log(1/\epsilon))$. This improves the previous
$O(n+\log(1/\epsilon))$ bound and establishes rapid
mixing beyond high-temperature and weak-interaction regimes. Our proof explains why the slowest decay rate, captured by the spectral
gap, can substantially underestimate the speed of equilibration.  The key
insight is that a state far from equilibrium must have a
substantial component outside the span of low-weight Pauli operators,
where stronger dissipation estimates apply. This yields two stages of
relaxation: a double-exponential convergence from a bad initial state,
followed by exponential convergence near equilibrium. We turn this
mechanism into a framework for proving rapid mixing through spectral
analysis beyond the gap.
Our decay estimate is equivalent to exponential convergence in
sandwiched R\'enyi-2 divergence. Together with the hypercontractivity framework of~\cite{GW2026},
this estimate yields a modified logarithmic Sobolev inequality with rate
$\Omega(\log^{-4}(n))$.
\end{abstract}

\section{Introduction}\label{sec:introduction}
Understanding quantum many-body systems at finite temperature is a central problem in quantum information science, condensed matter physics, and statistical mechanics. Thermal equilibrium states describe a wide range of experimentally relevant systems and encode fundamental physical properties, including correlations, thermodynamic observables, and finite-temperature phases~\cite{Alhambra2023}. Efficiently preparing and manipulating these states is therefore an important task in quantum simulation. For a system with Hamiltonian \(H\) at inverse temperature \(\beta\), the equilibrium state is described by the Gibbs state $\rho_\beta=\frac{e^{-\beta H}}{\operatorname{Tr}(e^{-\beta H})}$, and its preparation is commonly referred to as quantum Gibbs sampling~\cite{PoulinWocjan2009,ChowdhurySomma2017,Alhambra2023}. Besides providing direct access to finite-temperature expectation values, efficient Gibbs-state preparation serves as an important primitive for estimating thermodynamic quantities such as partition functions and free energies~\cite{PoulinWocjan2009,RFA2024rapid}.

Dissipative state preparation has been extensively studied in the context of quantum Gibbs sampling. The central idea is to engineer a dissipative quantum dynamics that can be simulated efficiently while rapidly driving the system toward the target Gibbs state~\cite{KB2016}. This approach is closely analogous to classical Markov-chain sampling, where a suitably designed stochastic dynamics converges to a prescribed stationary distribution. Early theoretical studies of dissipative Gibbs sampling largely relied on
the Davies generator, which describes a quantum system weakly coupled to a thermal
bath~\cite{Davies1974,Davies1976}. For non-commuting many-body
systems, however, the resulting dynamics can be nonlocal and
difficult to implement efficiently~\cite{CKBG2023}. Recent constructions address these
implementation difficulties with quasi-local Gibbs samplers that
satisfy Kubo--Martin--Schwinger (KMS) detailed balance and can be
simulated efficiently on quantum computers~\cite{CKG2023,DLL2024,CKBG2025}. Furthermore, several studies have established an upper bound \reviewedit{on} the convergence time, often referred to as the \emph{mixing time}, for different classes of Hamiltonians, including high-temperature or weakly interacting Hamiltonians~\cite{RFA2024rapid,SMBB2025rapid}, and one-dimensional local Hamiltonians~\cite{BC2025}. \reviewedit{Together, these results give} a theoretical guarantee for the end-to-end efficiency of dissipative Gibbs sampling for a wide range of physically relevant systems.

Although the efficient implementation of Davies generators can be challenging, their structure is particularly amenable to theoretical analysis of mixing times. For commuting Hamiltonians, the mixing behavior of Davies generators is by now relatively well understood, with general criteria linking the decay of equilibrium correlations to convergence rates~\cite{KB2016,KACR2025,StengeleEtAl2026CSS}. In particular, \emph{rapid mixing}, characterized by a mixing time that grows only polylogarithmically with the system size at fixed accuracy, has been established for one-dimensional commuting local spin systems at any fixed finite temperature~\cite{BCGLPR2024,KACR2025}, for higher-dimensional local spin systems at sufficiently high temperatures~\cite{KACR2025}, and for certain classes of quantum error-correcting-code Hamiltonians~\cite{StengeleEtAl2026CSS}. Extending these results to non-commuting Hamiltonians, however, presents significant difficulties. Rapid mixing bounds \reviewedit{for} the KMS Gibbs
samplers \cite{CKG2023,DLL2024} have only been established in
\reviewedit{perturbative regimes} such as high-temperature or weak-interaction regimes~\cite{RFA2024rapid,SMBB2025rapid}.
Beyond these \reviewedit{perturbative} regimes, even in one dimension, the quantum Gibbs sampler is only known to exhibit \emph{fast mixing}, with a recent
spectral-gap analysis providing a polynomial upper bound on the mixing
time~\cite{BC2025}. This raises a natural question:

\emph{Can one establish rapid mixing of the KMS detailed-balance Lindbladian for general one-dimensional non-commuting local spin systems at every fixed finite temperature? Or is there a fundamental obstruction to rapid mixing in this setting?}

We resolve this problem by establishing the first rapid mixing for general one-dimensional non-commuting local spin systems at every fixed
finite temperature. Specifically, we prove that, for a one-dimensional local spin system of $n$ sites with uniformly bounded nearest-neighbor interactions and a fixed local dimension, the KMS Gibbs sampler mixes to trace-norm error $\epsilon$ in time $\mathcal{O}(\mathrm{polylog}(n)+\log(1/\epsilon))$. An informal statement of our main result is as follows.

\begin{theorem}[Rapid mixing in one dimension, informal]\label{thm:main-1d}
For a one-dimensional spin system of $n$ sites with uniformly bounded nearest-neighbor
interactions and a fixed local dimension, at any fixed finite temperature and from any initial state,
the sampler mixes to trace-norm error $0<\epsilon<1$ in time
\[
 t_{\mathrm{mix}}(\epsilon)
 =O\!\left(\log^4(n)+\log(1/\epsilon)\right).
\]
\end{theorem}
The formal statement of Theorem~\ref{thm:main-1d} is given in Theorem~\ref{thm:main_result} of Section~\ref{sec:main-results}. In Theorem~\ref{thm:main-1d}, we consider the same Gibbs sampler as in~\cite{BC2025}, namely, a KMS detailed-balance Lindbladian with single-site Pauli coupling operators. For comparison,~\cite{BC2025} established the previously best-known mixing-time bound of $\mathcal{O}\!\left(n+\log(1/\epsilon)\right)$.

Conventionally, for Lindbladian dynamics satisfying KMS detailed balance, convergence in the $\chi^2$ divergence is controlled by the spectral gap with respect to the KMS inner product, which governs the slowest nontrivial decay mode. However, the worst-case initial $\chi^2$ divergence from equilibrium can grow exponentially with the system size. Consequently, even a system-size-independent spectral gap generally yields only a mixing-time bound that scales linearly with the system size. Establishing rapid mixing therefore typically requires a different measure of convergence that simultaneously admits a more favorable initial bound and a system-size-independent decay rate. To the best of our knowledge, three such approaches have been employed in the literature: (i) relative entropy~\cite{CRF2020,BCGLPR2024,KACR2025,StengeleEtAl2026CSS,StengeleEtAl2026QD}, (ii) oscillator norms~\cite{RFA2024rapid,SMBB2025rapid}, and (iii) quantum Wasserstein distances~\cite{BLMT2025,BT2026}. Relative entropy has proved particularly effective for establishing rapid mixing of Davies generators associated with commuting Hamiltonians. The key step is to establish a modified logarithmic Sobolev inequality (MLSI), which directly implies exponential decay of the relative entropy toward equilibrium. For noncommuting Hamiltonians and KMS detailed-balance Lindbladians, however, establishing suitable MLSI bounds remains challenging. Oscillator norms and quantum Wasserstein distances provide alternative routes to rapid mixing in the noncommuting setting. Existing results based on these approaches, however, typically require the influence between different sites to be sufficiently weak compared with the local relaxation rate, restricting their applicability to perturbative regimes such as sufficiently high temperatures or weak interactions.

In this work, we develop a new framework that goes beyond the three approaches described above. Unlike previous rapid-mixing analyses, our approach remains entirely within the conventional KMS $L^2$-space, namely, the Hilbert space of operators equipped with the KMS inner product. The key observation is that a state with a large $\chi^2$ divergence from the target state often has a substantial component in rapidly decaying spectral subspaces, whose decay rates can be much faster than the slowest rate determined by the spectral gap. When the system is far from equilibrium, these fast-decaying components can dominate the deviation from equilibrium, resulting in stronger effective dissipation and substantially faster convergence than what is captured by the spectral gap alone. This mechanism differs fundamentally from standard spectral-gap analyses, which uniformly control all initial states using only the slowest nontrivial decay rate.

To prove Theorem~\ref{thm:main-1d}, we provide sufficient conditions for rapid mixing of KMS detailed-balance Lindbladians in terms of the physical properties of the Gibbs state. We then apply this framework to general one-dimensional local spin
systems at every fixed finite temperature by rigorously showing that these conditions are satisfied. The condition verification process relies on well-established structural results in~\cite{BC2025}, including a comparison of Dirichlet forms, a size-independent global/local spectral gap, and exponential clustering in the Gibbs state. The framework can potentially be applied more generally to establish rapid mixing in other settings where the same structural conditions hold.


\subsection{Related work}\label{sec:related-work}

\paragraph*{Mixing times of quantum Gibbs samplers.}
The mixing-time analysis of dissipative quantum dynamics is most mature
for commuting Hamiltonians. A foundational insight in this regime
connects the spectral gaps of Davies generators to spatial correlations
in the underlying Gibbs state~\cite{KB2016}. Building on this connection,
modified logarithmic Sobolev inequality (MLSI) methods based on spatial
mixing and entropy factorization~\cite{BCLPR2021,CRF2020} have yielded
rapid mixing for one-dimensional commuting local spin systems at every fixed
finite temperature, as well as for higher-dimensional local spin systems at high
temperature~\cite{BCGLPR2023,BCGLPR2024,KACR2025}.
Further developments of MLSI methods include comparison bounds for
shallow-circuit parent Hamiltonians~\cite{BCL2024}, mixing criteria that
combine weak MLSI with transport inequalities~\cite{CGKR2025}, and
size-independent MLSI bounds for classes of two-dimensional quantum
code Hamiltonians, including the toric code and other Abelian quantum
double models, at all positive
temperatures~\cite{StengeleEtAl2026CSS,StengeleEtAl2026QD}.

For non-commuting Hamiltonians, the construction of quasi-local
samplers satisfying KMS detailed balance~\cite{CKG2023,DLL2024}
enabled a similar analytical program. First, in the perturbative regime of high temperature or weak interactions, meaning that the \reviewedit{inverse temperature} or interaction strength is smaller than a constant depending on the \reviewedit{local parameters} of the Hamiltonian but independent of the system size, several fast-mixing and rapid-mixing results have been established. At high temperature,
system size independent spectral gaps were
proved~\cite{RFA2024gap}, giving polynomial mixing times.
Rapid mixing in this regime was subsequently established using
oscillator norms~\cite{RFA2024rapid} and quantum Wasserstein
distances~\cite{BLMT2025}.
Building on the Wasserstein approach, \cite{BT2026} constructs a
sampler that mixes rapidly at high temperature even with arbitrary
on-site fields.
Weakly interacting systems provide another setting with fast-mixing
results for fermions~\cite{TZ2025,SMBB2025gap} and rapid-mixing results
for both spin and fermionic systems~\cite{SMBB2025rapid}.~Beyond these perturbative regimes, for general one-dimensional local spin systems, \cite{BC2025} proved a system-size independent spectral gap for the sampler of~\cite{CKG2023} at every fixed finite temperature. In summary, although rapid mixing has been established for commuting Hamiltonians and for non-commuting Hamiltonians in perturbative regimes, it remains an open question whether rapid mixing holds for general non-commuting Hamiltonians beyond these perturbative regimes. Our work resolves this question in the affirmative for one-dimensional local spin systems at every fixed finite temperature.

In a different direction, recent work has studied the spectral gaps of
Davies generators through comparison with classical chains~\cite{BGSSSV2025} and by
exploiting model symmetries, as in the mean-field Heisenberg
ferromagnet~\cite{BBLRS2026}.

\paragraph*{Algorithms for one-dimensional Gibbs states.}
It is well known that efficient classical and quantum procedures are
available for approximating and preparing one-dimensional Gibbs states
at any fixed finite temperature.
Classical tensor-network algorithms construct matrix product
approximations~\cite{MSVC2015}, with a quasilinear-time construction
in~\cite{KAA2021}. For translation-invariant chains, transfer-matrix
methods also give efficient classical algorithms for the thermodynamic
free energy and local Gibbs marginals~\cite{FFS2023}.
More recently, \cite{BCP2026} gives an exact decomposition into a mixture
of matrix product states with bond dimension independent of $n$, together
with a polynomial-time classical algorithm for approximately sampling
this ensemble.
On the quantum side, \cite{BK2019} constructs preparation circuits from
local recovery channels, and \cite{BC2025} obtains polylogarithmic-depth
preparation using an adiabatic construction.

\paragraph*{Connections to classical mixing theory.} Our framework is closely related to the spectral-profile method in classical Markov-chain mixing theory~\cite{GMT2006}. The classical spectral-profile method goes beyond the spectral gap by establishing stronger dissipation estimates for functions localized on sets of small stationary measure, leading to accelerated convergence when the distribution is far from equilibrium. Our work develops a quantum analogue of this underlying principle through a different decomposition in the KMS $L^2$ space. Specifically, we decompose operators into the subspace spanned by low-weight Pauli strings and its KMS-orthogonal complement. We establish stronger dissipation estimates for the complementary component, while using state positivity and normalization to control the contribution of the low-weight component. Together, these estimates show that states far from equilibrium necessarily experience stronger effective dissipation than that predicted by the spectral gap alone, providing a quantum counterpart of the enhanced-dissipation mechanism underlying classical spectral-profile methods.

\paragraph*{Concurrent work.}
Four concurrent works obtain further results for one-dimensional Gibbs
states at every fixed finite temperature.
\cite{PSBM2026} proves an exact decomposition into mixtures of product
states on contiguous blocks of bounded size. This gives constant-depth
local unitary preparation with polynomial-time classical preprocessing. \cite{Kim2026},~\cite{cps26}, and~\cite{rf26} proves $O(\log(n/\epsilon))$ mixing for a different
sampler: a block heat-bath dynamics whose updates trace out a contiguous
block of constant size and apply the Petz recovery map. Both \cite{Kim2026} and \cite{cps26} use
contraction in quantum Wasserstein distance and also yields
polylogarithmic-depth preparation circuits. In contrast, our work directly addresses the CKG sampler~\cite{CKG2023},
proving rapid mixing through spectral analysis in the KMS $L^2$ space. A more detailed comparison of our work with these concurrent works is given in the next subsection. After the first version of this work was posted on arXiv, Gao, Guo, and Wang~\cite{GGW2026} proved a complete modified logarithmic Sobolev inequality with a system-size-independent constant for the CKG sampler on one-dimensional spin chains with bounded finite-range interactions at every fixed positive temperature, yielding $O(\log n+\log(1/\epsilon))$ mixing in both trace norm and diamond norm. This gives a stronger mixing-time bound than our result, but the proof is based on a different approach that does not rely on spectral analysis in the KMS $L^2$ space.

\subsection{Implications}

Although our main result focuses on one-dimensional systems, our spectral analysis captures a general physical mechanism for rapid equilibration. The resulting framework can potentially be applied to a broader range of systems and related problems. We highlight three implications below.

\paragraph{A new and general mechanism for rapid mixing.} Our framework provides a simple approach to rapid mixing. Unlike approaches based on relative entropy, oscillation norms, or quantum Wasserstein distance, it relies only on spectral analysis in the KMS $L^2$ space. This also gives a direct physical interpretation of the relaxation process: faster-decaying modes drive the rapid initial approach to equilibrium, while the slow modes govern the final relaxation. The enhanced convergence thus reflects how different parts of the spectrum contribute at different stages of the dynamics.

For mixing-time analysis, this approach also allows us to treat more general dynamics because our dissipation estimates rely on a global Dirichlet-form comparison rather than local estimates like bounds requiring inter-site influence to be weak relative to on-site relaxation. For example, we can include a coherent term $-i[H,\cdot]$ in the dynamics without changing our dissipation estimates, since the corresponding unitary evolution preserves the Gibbs state and the KMS norm. Moreover, the comparison is not intrinsically tied to one-dimensional geometry, so the framework can potentially apply to more general Hamiltonians, including nonlocal systems such as the SYK model~\cite{CZZ2026}.

\paragraph{Rapid mixing of CKG-type Lindbladian beyond perturbative regimes.} Our result establishes rapid mixing of the CKG sampler~\cite{CKG2023} for general one-dimensional noncommuting local spin systems at any fixed finite temperature, extending beyond the previously established perturbative regimes of high temperature or weak interactions~\cite{RFA2024rapid,SMBB2025rapid}. A key practical advantage of the CKG sampler is that it can be implemented directly from the Hamiltonian and simple coupling operators, without requiring prior knowledge of the target Gibbs state. In contrast, other rapidly mixing samplers, such as Petz-recovery-based methods~\cite{Kim2026,cps26,rf26}, often require additional information about the Gibbs state, such as the ability to prepare local Gibbs marginals~\cite{Kim2026}.

Our work demonstrates that the robustness and simplicity of the CKG sampler do not come at the expense of slower convergence. In particular, we resolve the previously open question of whether the CKG sampler admits rigorous rapid-mixing guarantees beyond perturbative regimes.

\paragraph{Connections to other mixing metric.} The result developed in this work is also connected with other mixing metrics and classical mixing theory. Our decay estimate can also be expressed as exponential convergence in sandwiched R\'enyi-2 divergence. Together with the hypercontractivity framework of~\cite{GW2026}, this estimate yields a modified logarithmic Sobolev inequality with rate $\Omega(\log^{-4}(n))$. We will discuss these connections in Appendix~\ref{app:renyi-log-sobolev}. Thus, our result provides a new route to establishing MLSI for KMS-DBC Lindbladians, which is often difficult to prove directly.



\section{Preliminaries}\label{sec:kms-framework}
In this section, we introduce the class of Hamiltonians considered in this work, review Kubo-Martin-Schwinger (KMS) detailed-balance Lindbladian dynamics, and discuss the standard mixing-time analysis for such dynamics.

\begingroup
\paragraph*{Notation}
Let $\mathcal H$ be a finite-dimensional Hilbert space and let $\rho$
be a full-rank density operator on $\mathcal H$. We write
$\mathcal B(\mathcal H)$ for the space of \reviewedit{bounded} linear operators on $\mathcal H$,
and $\norm{\cdot}_1$, $\norm{\cdot}_2$, and $\norm{\cdot}_\infty$ for the
trace, Hilbert--Schmidt, and operator norms, respectively.
For a system of $n$ sites, each containing a fixed number $q$ of qubits,
the local dimension is $d=2^q$ and $\mathcal H=(\mathbb C^d)^{\otimes n}$.
For $A\subseteq[n]:=\{1,\ldots,n\}$, let
$\mathcal H_A=\bigotimes_{i\in A}\mathbb C^d$ denote the Hilbert space
of the subsystem $A$, with operator space $\mathcal B(\mathcal H_A)$.

\endgroup
\subsection{One-dimensional local spin systems}\label{sec:hamiltonian-setup}

We consider a one-dimensional open chain of $n>1$ sites, with a Hamiltonian of finite interaction range and uniformly bounded interaction strength:
\begin{equation}\label{eqn:local_hamiltonian}
 H=\sum_{\substack{Z\subseteq[n]\\\operatorname{diam}(Z)\leq r}}h_Z,
 \qquad
 \sup_{i\in[n]}\sum_{i\in Z}\norm{h_Z}_\infty\leq J,
\end{equation}
where $h_Z$ is \reviewedit{a Hermitian operator supported on the sites in} $Z$, meaning that $h_Z=\mathbb{I}_{[n]\setminus Z}\otimes H_Z$, where $H_Z\in\mathcal B(\mathcal H_Z)$ is Hermitian. Here, both the interaction range $r$ and strength bound $J$ are independent of $n$.

For the purpose of this work, we will only consider the nearest-neighbor interaction model with unit-strength interactions.
\begin{equation}\label{eqn:near_neighbor_interaction}
 H=\sum_{i=1}^{n-1}h_{i,i+1},
 \qquad \norm{h_{i,i+1}}_\infty\leq1,
\end{equation}
which corresponds to \reviewedit{the case} with $r=1$ and $J\leq2$. According to~\cite[Section I.1]{BC2025}, we can group neighboring sites and rescale the interaction strength to reduce~\eqref{eqn:local_hamiltonian} to the nearest-neighbor interaction model~\eqref{eqn:near_neighbor_interaction} by appropriately redefining $q$, the local dimension $d$, and \reviewedit{inverse temperature} $\beta$. \reviewedit{In the rest of the paper}, all the parameters $q$, $d$, and $\beta$ are considered fixed and independent of $n$. Without specific mention, the constants \reviewedit{appearing} in all subsequent analysis depend on these fixed parameters.

\subsection{KMS detailed-balance Lindbladian dynamics}

For a Hamiltonian $H$ given in~\eqref{eqn:near_neighbor_interaction}, and inverse temperature $\beta>0$, we define the thermal state as $\rho_\beta=\frac{e^{-\beta H}}{\Tr(e^{-\beta H})}$. The Lindbladian dynamics that prepares this thermal state often \reviewedit{takes the form}
\[
\frac{\mathrm{d}\rho}{\mathrm{d}t} = \mathcal L(\rho):= \reviewedit{-i[G,\rho]} + \sum_{\alpha} \left( L_\alpha \rho L_\alpha^\dagger - \frac{1}{2} \{ L_\alpha^\dagger L_\alpha, \rho \} \right)\,,
\]
where \reviewedit{the coherent term} $G$ and \reviewedit{jump operators} $L_{\alpha}$ \reviewedit{are chosen to ensure} that $\mathcal{L}(\rho_\beta)=0$, meaning $\rho_\beta$ is a stationary state of the Lindbladian dynamics. One way to achieve this is through the KMS detailed-balance condition, which ensures that the dynamics is reversible with respect to the thermal state. We put the rigorous definition of the KMS detailed-balance condition in Section~\ref{sec:KMS_geometry_mixing}.

The first efficiently simulable KMS detailed-balance Lindbladian for general local Hamiltonians was introduced in~\cite{CKG2023}. It takes the form
\begin{equation}\label{eq:setup-gibbs-generator}
 \mathcal L(\rho)=-i[G,\rho]
 +\sum_i\int_{\mathbb R}\gamma(\omega)
 \left(J_i(\omega)\rho J^\dagger_i(\omega)
 -\tfrac12\{J_i(\omega)^\dagger J_i(\omega),\rho\}\right)d\omega.
\end{equation}
Here, the jump operators are defined by
\[
 J_i(\omega)=\frac1{\sqrt{2\pi}}
 \int_{\mathbb R} f_\beta(t)e^{-i\omega t}
 e^{itH}A_i e^{-itH}\,dt,
 \qquad
 f_\beta(t)=\frac{(2/\pi)^{1/4}}{\sqrt\beta}e^{-t^2/\beta^2},
\]
where the family of coupling operators $\{A_i\}$ is chosen to be closed under adjoints. The weight function is given by $\gamma(\omega)=e^{-(\beta\omega+1)^2/2}$, and the coherent correction $G$ is defined in the energy eigenbasis, $H|e\rangle=e|e\rangle$, through
\[
 \langle e|G|e'\rangle
 =\frac{i}{2}\tanh\!\left(\frac{\beta(e-e')}{4}\right)
 \sum_i\int_{\mathbb R}\gamma(\omega)
 \langle e|J_i(\omega)^\dagger J_i(\omega)|e'\rangle\,d\omega.
\]
Equivalently, writing $A_i(t):=e^{itH}A_i e^{-itH}$,
the time-domain formula of~\cite[Corollary III.1]{CKG2023} gives
\[
 G=\sum_i\int_{\mathbb R}b_1(t)e^{-i\beta Ht}
       \left(\int_{\mathbb R}b_2(s)
       A_i(\beta s)^\dagger A_i(-\beta s)\,ds\right)
       e^{i\beta Ht}\,dt,
\]
where the smooth, rapidly decaying kernels are
\[
 \begin{aligned}
 b_1(t)=-2\sqrt\pi\,e^{1/8}
       \int_{\mathbb R}\frac{\sin(s)e^{-2s^2}}
                                  {\cosh(2\pi(t-s))}\,ds,\quad b_2(s)=\pi^{-3/2}e^{-4s^2-2is}.
 \end{aligned}
\]
Compared with Davies generators, the jump operators in~\eqref{eq:setup-gibbs-generator} are much easier to block encode since $f_\beta(t)$ is a smooth, rapidly decaying function. This ensures the efficient simulation of the corresponding Lindbladian dynamics on a quantum computer with \reviewedit{implementation details given} in~\cite{CKG2023}.

Later, in~\cite[Section 3.1]{DLL2024}, the authors introduced a more general construction of KMS detailed-balance Lindbladians using a finite family of jump operators and flexible filter functions. It takes the form
\begin{equation}\label{eq:setup-general-gibbs-generator}
 \mathcal L(\rho)=-i[G,\rho]
 +\sum_i\left(L_i\rho L_i^\dagger
 -\tfrac12\{L_i^\dagger L_i,\rho\}\right).
\end{equation}
Here, the jump operators are defined by
\[
 L_i=\int_{\mathbb R}f_i(t)e^{itH}A_i e^{-itH}\,dt,\quad f_i(t)=\frac1{2\pi}\int_{\mathbb R}
 q_i(\nu)e^{-\beta\nu/4}e^{-i\nu t}\,d\nu,
\]
where the coupling operator $A_i$ is self-adjoint and the weighting functions satisfy $q_i(-\nu)=\overline{q_i(\nu)}$. The filters are chosen so that $q_i(\nu)e^{-\beta\nu/4}$ and $f_i(t)$ are integrable. The coherent correction $G$ is defined in the energy eigenbasis through $\langle e|G|e'\rangle
 =\frac{i}{2}\tanh\!\left(\frac{\beta(e-e')}{4}\right)
 \sum_i\langle e|L_i^\dagger L_i|e'\rangle$.

In our work, we will only focus on the case of~\eqref{eq:setup-gibbs-generator} to illustrate the main ideas and results. We consider the same setting as in~\cite{BC2025} with $\{A_i\}$ being the set of all single-site Pauli operators. Specifically, let
\[
 \mathcal P_q=\{I_2,\sigma^x,\sigma^y,\sigma^z\}^{\otimes q}
 =\{P_u:0\leq u\leq d^2-1\},\qquad P_0=I_d,
\]
be the Hermitian Pauli basis of one site, where $\sigma^x,\sigma^y,\sigma^z$ are the usual single-qubit Pauli matrices. We choose the coupling family
\begin{equation}\label{eq:setup-pauli-couplings}
 \{A_{i,u}:1\leq i\leq n,\ 1\leq u\leq d^2-1\},\quad A_{i,u}=\sigma_{i,u}
 :=I_d^{\otimes(i-1)}\otimes P_u\otimes I_d^{\otimes(n-i)}.
\end{equation}
Here $i$ labels the site and $u$ labels its $d^2-1=4^q-1$ nonidentity Pauli operators. Each coupling has operator norm one. The coupling index in~\eqref{eq:setup-gibbs-generator} is thus the pair $(i,u)$. The extension of our analysis to other classes of Lindblad generators will be discussed later after Theorem~\ref{thm:main_result}.

\subsection{KMS geometry, spectral gap, and mixing time analysis}\label{sec:KMS_geometry_mixing}

\paragraph{KMS inner product and KMS detailed-balance condition.}

The Kubo--Martin--Schwinger (KMS) inner product and its associated norm are
\[
 \langle X,Y\rangle_\rho=\Tr(X^\dagger\rho^{1/2}Y\rho^{1/2}),
 \qquad
 \norm X_\rho=\norm{\rho^{1/4}X\rho^{1/4}}_2.
\]
For a Schr\"odinger-picture Lindblad generator $\mathcal L$, let
$\mathcal L^\dagger$ denote its Hilbert--Schmidt adjoint.
We say that $\mathcal L$ satisfies
\emph{KMS detailed balance} with respect to $\rho$ if $\mathcal{L}^\dagger$ is self-adjoint with respect to the KMS inner product, i.e.,
\[
 \langle X,\mathcal L^\dagger Y\rangle_\rho
 =\langle\mathcal L^\dagger X,Y\rangle_\rho
 \qquad\text{for all }X,Y\in\mathcal B(\mathcal H)\,.
\]
This ensures the stationarity of $\rho$, meaning that $\mathcal L(\rho)=0$.

\paragraph{Dirichlet form and \reviewedit{spectral} gap.}

For a KMS detailed-balanced $\mathcal L$, the operator $-\mathcal L^\dagger$ is self-adjoint and nonnegative with respect to the KMS inner product. Define its Dirichlet form by
\begin{equation}\label{eqn:Dirichlet_form}
  \mathcal E_{\mathcal L}(X,Y):=\langle X,-\mathcal L^\dagger Y\rangle_\rho\,.
\end{equation}
If $\rho$ is the unique stationary state, the spectral gap is given by
\begin{equation}\label{eq:setup-gap}
 \lambda_{\rm gap}\left(\mathcal{L}\right)=\inf_{\substack{X\ne0\\\Tr(\rho X)=0}}
 \frac{\mathcal E_{\mathcal L}(X,X)}{\norm X_\rho^2}>0.
\end{equation}
Equivalently, the spectral gap $\lambda_{\rm gap}\left(\mathcal{L}\right)$ is the largest constant such that $\mathcal E_{\mathcal L}(X,X)\geq\lambda_{\rm gap}\left(\mathcal{L}\right)\norm X_\rho^2$ whenever $\Tr(\rho X)=0$.

\paragraph{Mixing time from the spectral gap.}

For a KMS detailed-balanced $\mathcal L$ with unique stationary
state $\rho$, define the worst-case mixing time at accuracy
$0<\epsilon<1$ by
\[
 t_{\mathrm{mix}}(\epsilon;\mathcal{L})
 :=\inf\left\{t\geq0:
 \sup_{\eta_0\ \mathrm{density}}
 \norm{e^{t\mathcal L}(\eta_0)-\rho}_1\leq\epsilon\right\}.
\]
There is a direct connection between the spectral gap and the mixing time, as the spectral gap controls the rate of exponential decay of \reviewedit{the squared KMS norm $F$ defined below}.
For any $A$ with $\Tr(\rho A)=0$, let
$A_t=e^{t\mathcal L^\dagger}A$ and $F(t)=\norm{A_t}_\rho^2$.
The gap \eqref{eq:setup-gap} implies exponential decay:
\[
 F(t)\leq e^{-2\lambda_{\rm gap}\left(\mathcal{L}\right) t}F(0).
\]
Applying this to the state evolution $\eta_t=e^{t\mathcal L}(\eta_0)$,
with $A=\rho^{-1/2}(\eta_0-\rho)\rho^{-1/2}$,
we obtain $F(t)=\chi^2(\eta_t\Vert\rho)$, where
\[
 \chi^2(\eta\Vert\rho)
 :=\norm{\rho^{-1/4}(\eta-\rho)\rho^{-1/4}}_2^2
\]
is the KMS $\chi^2$ divergence of a state $\eta$ from $\rho$.
The bounds $\norm{\eta_t-\rho}_1^2\leq F(t)$ and
$F(0)\leq\lambda_{\min}(\rho)^{-1}-1$ therefore yield
\begin{equation}\label{eq:setup-gap-mixing}
 t_{\mathrm{mix}}(\epsilon;\mathcal{L})
 \leq\frac{1}{2\lambda_{\rm gap}\left(\mathcal{L}\right)}\log\frac1{\lambda_{\min}(\rho)}
       +\frac1{\lambda_{\rm gap}\left(\mathcal{L}\right)}\log\frac1\epsilon.
\end{equation}
\reviewedit{For the Gibbs state $\rho_\beta$ of the Hamiltonian in~\eqref{eqn:near_neighbor_interaction}, the bound $\log(\lambda_{\min}(\rho_\beta)^{-1})\leq n\log d+2\beta(n-1)$ gives a mixing-time upper bound that is linear in $n$ when the spectral gap is independent of the system size.} In~\cite{BC2025}, the authors showed that $\lambda_{\rm gap}\left(\mathcal{L}\right)=\Omega(1)$ in the case of~\eqref{eq:setup-gibbs-generator} and~\eqref{eq:setup-pauli-couplings}. According to the analysis above, this implies that the mixing time $t_{\mathrm{mix}}(\epsilon;\mathcal{L})=\mathcal{O}(n+\log\frac1\epsilon)$.

\section{Main result}\label{sec:main-results}
The main contribution of our work is to improve $n$ to $\operatorname{polylog}(n)$ in the mixing time bound, as stated in the following theorem.
\begin{theorem}\label{thm:main_result}
Assume $H$ takes the form of~\eqref{eqn:near_neighbor_interaction} {and $n>20$}. Given $\beta>0$ and $\mathcal{L}$ as in~\eqref{eq:setup-gibbs-generator} with~\eqref{eq:setup-pauli-couplings}, we have
\[
 t_{\mathrm{mix}}(\epsilon;\mathcal{L})=\mathcal{O}\left(\log^4(n)+\log\frac1\epsilon\right).
\]
Here, the constant in the $\mathcal{O}$ notation is independent of $n$ and $\epsilon$ but may depend on the fixed model parameters such as $\beta,q,d$.
\end{theorem}
We will discuss the proof of Theorem~\ref{thm:main_result} in Section~\ref{sec:spectral-criterion}.

As we will \reviewedit{discuss} in Section~\ref{sec:spectral-criterion}, the proof is different from conventional approaches, such as \reviewedit{relative entropy, oscillator norms, and Wasserstein distances}, for proving mixing time bounds. Instead, our proof starts from the
Dirichlet-form comparison and size-independent spectral gap
established in~\cite{BC2025}. The main step is to obtain stronger
dissipation estimates for
$A=\rho_\beta^{-1/2}(\eta-\rho_\beta)\rho_\beta^{-1/2}$ when
$F=\|A\|_{\rho_\beta}^2$ is large. We decompose $A$ into its KMS orthogonal projection onto the
span of low-weight Pauli strings and the complementary component.
We prove that the KMS norm of the complementary component
is controlled by the Dirichlet form, with a coefficient that
decreases as the weight cutoff increases, up to logarithmic
losses in the system size. Meanwhile, state normalization and
an operator-norm estimate on the low-weight subspace bound the
contribution of the low-weight component. Consequently, when
$F$ is large, a substantial fraction of its norm must lie outside
this subspace and undergo stronger dissipation. Combining these estimates gives the enhanced decay described in
Section~\ref{sec:spectral-criterion}, reducing the worst-case
exponentially large initial $F$ to order one in
$O(\log^4(2n))$ time. The size-independent spectral gap then
gives trace-norm accuracy $\epsilon$ in an additional
$O(\log(1/\epsilon))$ time. We summarize the proof strategy in Figure~\ref{fig:proof-strategy}.
\begin{figure}[!htbp]
\centering
\includegraphics[height=0.36\textheight,width=1\textwidth]{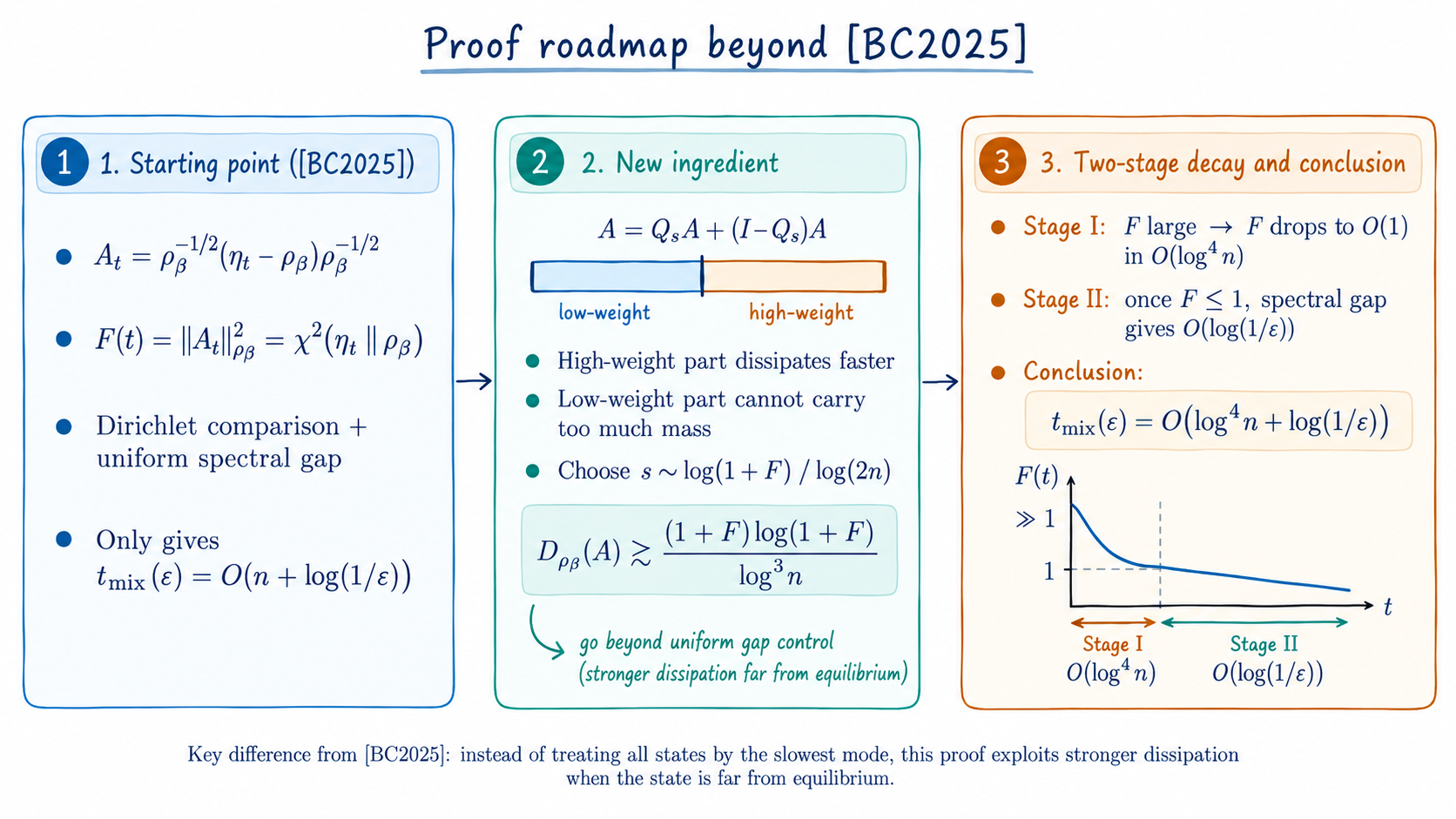}
\caption{Overview of the proof strategy for rapid mixing.}
\label{fig:proof-strategy}
\end{figure}

\begin{rem}[Extension to other Lindblad dynamics] In our analysis, the only step that relies on the specific form of the Lindblad generator is the Dirichlet-form comparison in Condition~\ref{cond:framework-comparison}. All the remaining analysis is only a property of thermal state $\rho_\beta$. Thus, for other Lindblad dynamics, as long as a similar comparison can be established, our framework can be applied to derive rapid mixing bounds. According to~\cite[Remark VI.1]{BC2025}, this might include the general class of KMS detailed-balanced Lindblad dynamics in~\eqref{eq:setup-general-gibbs-generator} when $f$ is properly chosen.
\end{rem}

\subsection{Sufficient conditions for rapid mixing}
\label{sec:spectral-criterion}
In this section, we will present \reviewedit{two known estimates} and one theorem that we will use to prove Theorem~\ref{thm:main_result}. \reviewedit{These estimates} were already established in~\cite{BC2025} but we will restate them here for completeness. We will discuss the proof of Theorem~\ref{thm:framework-enhanced-dissipation} in Section~\ref{sec:pf-enhanced-dissipation} and establish several technical lemmas that will be proved in Section~\ref{sec:proof-of-lemma}.

In proving \reviewedit{a lower bound on the spectral gap} in~\cite{BC2025}, one key observation is that the Dirichlet form of the Lindblad generator $\mathcal{L}$ can be compared with a simpler KMS \reviewedit{self-adjoint} superoperator $\mathcal{K}$, whose gap is much easier to analyze. The starting point of our analysis will also be this comparison, which is stated in Fact~\ref{cond:framework-comparison} below.
\begin{condition}[Dirichlet-form comparison]\label{cond:framework-comparison}
Consider $\mathcal L$ as in~\eqref{eq:setup-gibbs-generator} with~\eqref{eq:setup-pauli-couplings}. There exists a constant $c_{\mathrm{cmp}}>0$ that is independent of system size $n$ such that
\begin{equation}\label{eq:framework-comparison}
 \ip X{-\mathcal L^\dagger X}\geq
 c_{\mathrm{cmp}}\mathcal D_{\rho_\beta}(X),
 \qquad
 \mathcal D_{\rho_\beta}(X):=\sum_{i,u}\norm{[\sigma_{i,u},X]}_{\rho_\beta}^2
\end{equation}
for every Hermitian $X$\,. Furthermore, there exists a constant $c_{\mathrm{gap}}>0$ that is independent of system size $n$ such that
\begin{equation}\label{eq:framework-gap}
 \mathcal D_{\rho_\beta}(X)\geq c_{\mathrm{gap}}\norm X_{\rho_\beta}^2
\end{equation}
for every Hermitian $X$ with $\Tr(\rho_\beta X)=0$.
\end{condition}
This fact is proved in~\cite[Theorems V.1,VI.1]{BC2025}. Here, we can rewrite
\[
\mathcal D_{\rho_\beta}(X)=-\left\langle X,\mathcal{K}(X)\right\rangle_{\rho_\beta},\quad \mathcal{K}(X):=\sum_{i,u}\Bigl([\sigma_{i,u},X]\left(\rho_\beta^{1/2} \sigma_{i,u}^\dagger\rho_\beta^{-1/2}\right)-\left(\rho_\beta^{-1/2} \sigma_{i,u}^\dagger\rho_\beta^{1/2}\right)[\sigma_{i,u},X]\Bigr)\,,
\]
where $\mathcal{K}$ is a KMS self-adjoint superoperator. The comparison \eqref{eq:framework-comparison} reduces the desired spectral gap estimate in~\eqref{eq:setup-gap} to a lower bound on $\mathcal D_{\rho_\beta}(A)$. Different from the original Dirichlet form in~\reviewedit{\eqref{eqn:Dirichlet_form}}, the commutator form $\mathcal D_{\rho_\beta}(A)$ involves exactly local Pauli couplings, which is more convenient for the analysis. After establishing the \reviewedit{comparison},~\cite{BC2025} \reviewedit{proves} strong clustering for spectral conditional expectations, obtains an inverse-polylogarithmic conditional gap for $\mathcal{K}$ through a recursive decomposition, and then bootstraps it to a constant gap \reviewedit{for} $\mathcal{K}$. This gives a size-independent spectral gap for $\mathcal{L}$ and thus a linear mixing time \reviewedit{upper} bound in $n$.

In our work, we will heavily use the commutator form $\mathcal D_{\rho_\beta}(A)$ and conduct a domain decomposition in the Pauli basis to obtain a stronger dissipation estimate than~\eqref{eq:framework-gap} in~\cite{BC2025}. Instead, we will prove a logarithmic enhancement in the dissipation estimate, which will lead to a polylogarithmic mixing time bound in $n$. This is summarized in the following theorem:
\begin{theorem}[Enhanced dissipation estimate]\label{thm:framework-enhanced-dissipation} \reviewedit{Let $H$ be as in~\eqref{eqn:near_neighbor_interaction}, let $\beta>0$, and let $n\geq20$.} There exists a constant $c_{\mathrm{enh}}>0$ that is independent of system size $n$ such that
\begin{equation}\label{eq:framework-enhanced-dissipation}
 \mathcal D_{\rho_\beta}(A)\geq c_{\mathrm{enh}}\frac{\left(1+\norm A_{\rho_\beta}^2\right)\log \left(1+\norm A_{\rho_\beta}^2\right)}{\log^3(n)}
\end{equation}
for every Hermitian operator $A$ with $\Tr(\rho_\beta A)=0$ and $I+A\geq0$, equivalently $A=\rho_\beta^{-1/2}(\eta-\rho_\beta)\rho_\beta^{-1/2}$ for a density operator $\eta$, where $c_{\mathrm{enh}}$ may depend on $\beta,d,q$ but is independent of $n$.
\end{theorem}
  We note that the lower bound~\eqref{eq:framework-enhanced-dissipation} is not always stronger than the original one in~\cite{BC2025}. Although we have extra $\log(1+\norm A_{\rho_\beta}^2)$ factor in the lower bound, the denominator $\log^3(n)$ increases in $n$. In our proof, we will mainly use this enhanced dissipation estimate in the regime of large $\norm A_{\rho_\beta}^2$, which corresponds to the initial state $\eta_0$ that is far away from the thermal state $\rho_\beta$. In this regime, the logarithmic factor $\log(1+\norm A_{\rho_\beta}^2)$ will dominate the denominator $\log^3(n)$, which gives a stronger dissipation estimate than the original one in~\cite{BC2025}. This will be crucial for proving rapid mixing in $O(\log^4(n))$ time.

We now show that Theorem \ref{thm:framework-enhanced-dissipation} suffices for proving Theorem \ref{thm:main_result}:
\begin{proof}[Proof of Theorem~\ref{thm:main_result}]
We will use both the enhanced dissipation estimate in Theorem~\ref{thm:framework-enhanced-dissipation} and the size-independent spectral gap in Condition~\ref{cond:framework-comparison} to prove rapid mixing. Let $\eta_t=e^{t\mathcal L}(\eta_0)$ be the state evolution under the Lindblad dynamics with initial state $\eta_0$. Define $A=\rho_\beta^{-1/2}(\eta_0-\rho_\beta)\rho_\beta^{-1/2}$ and $F(t)=\norm{A_t}_{\rho_\beta}^2$ with $A_t=e^{t\mathcal L^\dagger}(A)$. Then, we have
\[
\left\|\exp(t\mathcal L)(\eta_0-\rho_\beta)\right\|_1^2\leq F(t)=\norm{A_t}_{\rho_\beta}^2.
\]
Thus, it suffices to prove that $F(t)\leq\epsilon^2$ after $O(\log^4(n)+\log(1/\epsilon))$ time.

According to Theorem~\ref{thm:framework-enhanced-dissipation}, we have
\[
 \frac{d}{dt}\log(1+F(t))=-\frac{2\mathcal E_{\mathcal L}(A_t,A_t)}{1+F(t)}
 \leq-\frac{2c_{\mathrm{cmp}}c_{\rm enh}}{\log^3(n)}\log(1+F(t)).
\]
This implies that
\[
F(t)\leq \exp\left(e^{-2c_{\rm cmp}c_{\rm enh}t/\log^3(n)}\log(1+F(0))\right)-1.
\]
Because $F(0)\leq\lambda_{\min}(\rho_\beta)^{-1}-1=\exp(\mathcal{O}( n))$, we have
\[
F(t)\leq
\exp\!\left(\mathcal{O}(n)\,e^{-2c_{\rm cmp}c_{\rm enh}t/\log^3(n)}\right)-1.
\]
Thus, to achieve $\epsilon$ accuracy, it suffices to choose $t=\mathcal{O}(\log^4(n)+\log^3(n)\log(1/\epsilon))$.

Next, to get rid of $\log^3(n)$ factor in the second term, we can use the size-independent spectral gap in Condition~\ref{cond:framework-comparison} in the case when $F(t)$ is small. Specifically, we can choose $t_0=\mathcal{O}(\log^4(n))$ such that $F(t_0)\leq 1$. Then, for $t\geq t_0$, we have
\[
\frac{d}{dt}F(t)=-2\mathcal E_{\mathcal L}(A_t,A_t)\leq -2c_{\rm cmp}c_{\rm gap}F(t),
\]
which implies that $F(t)\leq F(t_0)\exp(-2c_{\rm cmp}c_{\rm gap}(t-t_0))\leq \exp(-2c_{\rm cmp}c_{\rm gap}(t-t_0))$. Thus, to achieve $\epsilon$ accuracy, it suffices to choose $t-t_0=\mathcal{O}(\log(1/\epsilon))$. Combining the two time intervals, we have $t=\mathcal{O}(\log^4(n)+\log(1/\epsilon))$ to achieve $\epsilon$ accuracy, which completes the proof.
\end{proof}

\subsection{Proof of Theorem~\ref{thm:framework-enhanced-dissipation}}\label{sec:pf-enhanced-dissipation}
In this section, we prove Theorem~\ref{thm:framework-enhanced-dissipation}. Let $F=\|A\|^2_{\rho_\beta}$. Because we already have the size-independent spectral gap in Condition~\ref{cond:framework-comparison}, it suffices to prove a stronger dissipation estimate for large $F$, where the target is to show that
\begin{equation}\label{eq:framework-target}
 \mathcal D_{\rho_\beta}(A)\gtrsim\frac{F\log F}{\log^3(n)},
 \qquad F=\norm A_{\rho_\beta}^2=\chi^2(\eta\Vert\rho_\beta),
\end{equation}
where $A=\rho_\beta^{-1/2}(\eta-\rho_\beta)\rho_\beta^{-1/2}$ for a density operator
$\eta$.

\paragraph*{Low-weight and high-weight contributions.}
To prove \eqref{eq:framework-target}, we decompose $A$ into a
low-weight part and a high-weight part in the Pauli basis.
For an integer $s\geq0$, define
\[
 V_{\le s}:=\operatorname{span}\{P:\ P\text{ is a Pauli string with }w(P)\leq s\},
\]
where $w(P)$ counts the sites on which $P$ has a nonidentity factor, defined as
\[
w(P)=\sum_{i=1}^n\mathbf{1}_{\{P_i\neq I_d\}},\qquad P=\bigotimes_{i=1}^n P_i.
\]
Let $Q_s$ be the KMS orthogonal projection onto $V_{\le s}$.
Orthogonality gives
\begin{equation}\label{eq:framework-weight-split}
 A=Q_sA+(I-Q_s)A,
 \qquad
 F=\norm{Q_sA}_{\rho_\beta}^2+\norm{(I-Q_s)A}_{\rho_\beta}^2.
\end{equation}
The intuition is that high-weight operators might produce large $F$ but decay quickly, whereas
low-weight operators form a sparse part of the operator basis and thus cannot contribute too much to $F$. We use these properties in complementary ways.
\begin{itemize}
\item \emph{High-weight part.} Recall that $\mathcal D_{\rho_\beta}$ in~\eqref{eq:framework-comparison} is a sum of squared KMS norms of commutators. When a Pauli string $P$ has large weight $w(P)$, there are many sites $i$ with a nonzero commutator $[\sigma_{i,u},P]$ for some $u$, giving many nonzero contributions to $\mathcal D_{\rho_\beta}(P)$. More specifically, we anticipate that $\mathcal{D}_{\rho_\beta}(P)$ is proportional to $w(P)$. At infinite temperature, this relation is exact: Pauli strings are eigenmodes
of $\mathcal K$, with decay rates proportional to their weight.
For the quasi-local dynamics at finite temperature, we seek an
analogous relation between weight and dissipation, allowing for
a loss that scales \reviewedit{polylogarithmically} in $n$. Specifically, we will show the following lemma:
\begin{lemma}[Weight-dissipation relation]\label{lem:framework-weight-dissipation}
Under the conditions of Theorem~\ref{thm:main_result} and $n>1$, there exists a constant $c_{\rm weight}>0$ that is independent of system size $n$ such that
\begin{equation}\label{eq:framework-weight-energy}
 \mathcal D_{\rho_\beta}(A)
 \geq c_{\rm weight}\frac{s+1}{\log^2(n)}\norm{(I-Q_s)A}_{\rho_\beta}^2
 =c_{\rm weight}\frac{s+1}{\log^2(n)}\bigl(F-\norm{Q_sA}_{\rho_\beta}^2\bigr)\,.
\end{equation}
\end{lemma}
According to Lemma~\ref{lem:framework-weight-dissipation}, if $A$ has a large high-weight component, meaning $F-\norm{Q_sA}_{\rho_\beta}^2\approx F$, then \reviewedit{the ratio satisfies} $\mathcal D_{\rho_\beta}(A)/\|A\|^2_{\rho_\beta}\gtrsim \frac{s}{\log^2(n)}$, which can be much larger than the size-independent spectral gap in the case when $s\gg \log^2(n)$. This gives a stronger dissipation of the high-weight part in \eqref{eq:framework-weight-split}.

\item \emph{Low-weight part.}
For the low-weight part, we \reviewedit{do not expect} a strong dissipation estimate because the low-weight Pauli strings are sparse and thus cannot contribute too much to $F$. Instead, we will show that the small number of low-weight Pauli strings implies that the low-weight part cannot contribute too much to $F$, meaning that $\norm{Q_sA}_{\rho_\beta}^2$ can be well controlled when $s$ increases.

First, direct counting gives
$\dim V_{\le s}\leq\exp[O(s\log(2n))]$.
This sparsity motivates an operator-norm bound summarized in the following lemma:
\begin{lemma}[Low-weight operator-norm bound]\label{lem:framework-low-weight-norm}
Under the conditions of Theorem~\ref{thm:main_result}, there exists a constant $c_{\rm norm}>0$ that is independent of system size $n$ such that
\begin{equation}\label{eqn:framework-low-weight-norm}
 \norm Y_\infty\leq e^{c_{\rm norm}\log(2n)(s+1)}\norm Y_{\rho_\beta},\quad Y\in V_{\le s}\,.
\end{equation}
\end{lemma}

Such a bound directly controls the
low-weight contribution through state normalization:
\begin{equation}\label{eq:framework-low-weight-normalization}
 \norm{Q_sA}_{\rho_\beta}^2
 =\ip{Q_sA}A
 =\Tr\bigl((Q_sA)^\dagger(\eta-\rho_\beta)\bigr)
 \leq2\norm{Q_sA}_\infty.
\end{equation}
This therefore gives
\begin{equation}\label{eq:framework-low-weight-bound}
 \norm{Q_sA}_{\rho_\beta}^2\leq4e^{2c_{\rm norm}\log(2n)(s+1)}.
\end{equation}
Although this bound grows exponentially in $s$, \reviewedit{its logarithm grows linearly in $(s+1)\log(2n)$}. \reviewedit{Specifically, when $F$ is sufficiently large, choosing $s=\Theta(\log(F)/\log(2n))$ with a sufficiently small implicit constant makes the bound at most $F/2$}, while the high-weight estimate \eqref{eq:framework-weight-energy} gives the term $\frac{F\log(F)}{\log^3(n)}$.
\end{itemize}

We will put the proof of Lemma~\ref{lem:framework-weight-dissipation} and Lemma~\ref{lem:framework-low-weight-norm} in Section~\ref{sec:proof-of-lemma} and discuss the proof strategy in Section~\ref{sec:proof-strategy}.

Now we combine Lemma~\ref{lem:framework-weight-dissipation} and Lemma~\ref{lem:framework-low-weight-norm} to prove Theorem~\ref{thm:framework-enhanced-dissipation}.

\begin{proof}[Proof of Theorem~\ref{thm:framework-enhanced-dissipation}]
Write $F=\norm A_{\rho_\beta}^2$, and
$D=\mathcal D_{\rho_\beta}(A)$. The case $F=0$ is immediate, so assume
$F>0$. Combining~\eqref{eq:framework-weight-split},
\eqref{eq:framework-weight-energy}, and
\eqref{eq:framework-low-weight-bound} gives
\[
 F\leq4e^{2c_{\rm norm}\log(2n)(s+1)}
       +\frac{\log^2(n)}{c_{\rm weight}(s+1)}D
\]
for every integer $s\geq0$. Choose the cutoff so that the second
term is at most $F/2$:
\[
 s+1=\left\lceil\frac{2\log^2(n)D}{c_{\rm weight}F}\right\rceil
 \leq\frac{2\log^2(n)D}{c_{\rm weight}F}+1
 \leq\left(\frac{2\log^2(n)}{c_{\rm weight}}
           +\frac1{c_{\rm gap}}\right)\frac DF.
\]
The last inequality uses $D\geq c_{\rm gap}F$ from
\eqref{eq:framework-gap}. Substitution and rearrangement give
\[
 F\leq8\exp\!\left[
 2c_{\rm norm}\log(2n)
 \left(\frac{2\log^2(n)}{c_{\rm weight}}
       +\frac1{c_{\rm gap}}\right)\frac DF\right].
\]
Since the exponent is nonnegative, it follows that
\[
 \log(1+F)\leq\log9+
 2c_{\rm norm}\log(2n)
 \left(\frac{2\log^2(n)}{c_{\rm weight}}
       +\frac1{c_{\rm gap}}\right)\frac DF.
\]
Using $\log(1+F)\leq F$ and $F\leq D/c_{\rm gap}$, we obtain
\[
 \begin{aligned}
 (1+F)\log(1+F)
 &\leq F\bigl[1+\log(1+F)\bigr]\\
 &\leq\left[
 \frac{1+\log9}{c_{\rm gap}}
 +\frac{4c_{\rm norm}}{c_{\rm weight}}\log(2n)\log^2(n)
 +\frac{2c_{\rm norm}}{c_{\rm gap}}\log(2n)
 \right]D.
 \end{aligned}
\]
For $n\geq2$, we have $\log(2n)\leq2\log n$ and
$\log n\geq\log2$. Thus
\[
 \mathcal D_{\rho_\beta}(A)
 \geq c_{\rm enh}\frac{(1+F)\log(1+F)}{\log^3(n)},
 \qquad
 c_{\rm enh}:=\left(
 \frac{8c_{\rm norm}}{c_{\rm weight}}
 +\frac{4c_{\rm norm}}{c_{\rm gap}(\log2)^2}
 +\frac{1+\log9}{c_{\rm gap}(\log2)^3}
 \right)^{-1}.
\]
This coefficient is positive and independent of $n$.
\end{proof}

\section{Intuition and proof ideas for the technical lemmas}\label{sec:proof-strategy}

The two lemmas address different difficulties created by \reviewedit{Pauli weight} in the KMS inner product. Lemma~\ref{lem:framework-weight-dissipation}
shows high-weight operators dissipate quickly, whereas
Lemma~\ref{lem:framework-low-weight-norm} controls the operator norm of
an arbitrary linear combination of low-weight operators. The estimates we seek are
\[
 \norm{X-Q_sX}_{\rho_\beta}^2
 \leq\frac{C\log^2(2n)}{s+1}\mathcal D_{\rho_\beta}(X),
 \qquad
 \norm Y_\infty\leq(Cn)^{cs}\norm Y_{\rho_\beta}
 \quad(Y\in V_{\leq s}).
\]
In this section, we explain the constructions behind these estimates; the detailed
proofs appear in Section~\ref{sec:proof-of-lemma}.

\subsection{Proof of Lemma~\ref{lem:framework-weight-dissipation}: From local projections to a weight-dissipation bound}
\label{sec:oned-decomposition-sketch}

\paragraph*{An ideal commuting decomposition.} We first consider an idealized scenario where the system can be decomposed into commuting parts.  Assume that there exists a partition of the chain into disjoint intervals $U_1,\ldots,U_m$ and there are commuting KMS orthogonal
projections $\mathcal P_i$ such that
\begin{itemize}
\item (Condition 1): $\mathcal P_i$ is the KMS orthogonal projection onto the subspace of operators that act as the identity on $U_i$, namely $\ran\mathcal P_i=I_{U_i}\otimes\mathcal B(\mathcal H_{U_i^c})$.
\item (Condition 2): $\mathcal P_i$ commutes with the local projections for disjoint intervals, $[\mathcal P_i,\mathcal P_j]=0$ for $U_i\cap U_j=\emptyset$.
\item (Condition 3): $\mathcal P_i$ satisfies a local-gap estimate, meaning that for some constant $C$ independent of $n$,
\begin{equation}\label{eqn:local-gap}
 \norm{(I-\mathcal P_i)X}_{\rho_\beta}^2
 \leq C\mathcal D_{\rho_\beta,U_i}(X),
\end{equation}
where
\begin{equation}\label{eqn:local-dissipation}
\mathcal D_{\rho_\beta,R}(X):=
\sum_{x\in R,u}\norm{[\sigma_{x,u},X]}_{\rho_\beta}^2\,.
\end{equation}
\end{itemize}
These projections give us a natural decomposition of $X$ into components that are supported on different subsets of intervals:
\begin{equation}\label{eqn:decomposition}
 X=\sum_{S\subseteq[m]}X_S,\qquad
 X_S=\prod_{i\in S}(I-\mathcal P_i)
      \prod_{i\notin S}\mathcal P_iX,
 \qquad
 \operatorname{supp}(X_S)\subseteq\bigcup_{i\in S}U_i.
\end{equation}
Commutativity preserves all the imposed identity factors, and
orthogonality makes distinct components orthogonal. Thus
\begin{equation}\label{eqn:decomposition-norm}
\begin{aligned}
 \sum_S|S|\norm{X_S}_{\rho_\beta}^2=
 \sum_{i}\sum_{S \ni i}\norm{X_S}_{\rho_\beta}^2
 =\sum_i\norm{(I-\mathcal P_i)X}_{\rho_\beta}^2
 \leq C\mathcal D_{\rho_\beta}(X),
\end{aligned}
\end{equation}
where the equality uses $(I-\mathcal{P}_i)X_S= 1_{i \in S} X_S $ and the orthogonality. This implies that
\[
 \norm{\sum_{|S|>k}X_S}_{\rho_\beta}^2
 \leq\frac{C}{k+1}\mathcal D_{\rho_\beta}(X).
\]
Now, let's assume the length of each interval is bounded by $\ell$. Then, any component $X_S$ with $|S|\leq k$ has Pauli weight at most $k\ell$. Because $Q_{k\ell}$ is an \reviewedit{KMS orthogonal projection onto} the space of operators with Pauli weight at most $k\ell$, we have
\[
\left\|(I-Q_{k\ell})X\right\|_{\rho_\beta}^2\leq  \norm{\sum_{|S|>k}X_S}_{\rho_\beta}^2
 \leq\frac{C\ell}{(k+1)\ell}\mathcal D_{\rho_\beta}(X)\,.
\]
This gives a weight-dissipation bound with a polylogarithmic loss in $n$ if $\ell=\mathrm{polylog}n$.

Unfortunately, the idealized scenario above does not hold for the correlated Gibbs state. The third condition largely restricts the choice of $\reviewedit{\mathcal P_i}$. In~\cite{BC2025}, the authors construct a \reviewedit{KMS orthogonal projection} $\mathbb E_U$, called \reviewedit{a spectral conditional expectation}, that is KMS orthogonal and satisfies~\eqref{eqn:local-gap}. However, one can not directly use $\mathbb E_U$ because it can act nontrivially outside $U$, and projections for disjoint intervals need not commute. In particular, applying a second projection can destroy the identity factors imposed by the first. These difficulties can occur even for commuting interacting Hamiltonians.

To overcome this difficulty, we will introduce two steps of \reviewedit{domain} decompositions that involve a buffered region.

\paragraph{Step 1: Relaxing condition 1.}

In the \reviewedit{first} step of the separation, we will relax condition 1, guarantee condition 2 exactly, and preserve condition 3 approximately. Specifically, given an interval $U_i$, we will divide it into three parts $U_i=A_iB_iC_i$, where $B_i$ is the main part of the interval and $A_i,C_i$ are buffers, see Figure~\ref{fig:oned-interval-geometry}.
\begin{figure}[!htbp]
\centering
\begin{tikzpicture}[
 every node/.style={font=\small},
 side/.style={draw=orange!65!black,fill=orange!12,line width=.6pt},
 middle/.style={draw=blue!65!black,fill=blue!9,line width=.6pt}]
 \foreach \x/\lab/\sty in {
  0/B_1/middle,1.6/C_1/side,3.2/A_2/side,4.8/B_2/middle,
  6.4/C_2/side,8/A_3/side,9.6/B_3/middle}{
  \draw[\sty] (\x,0) rectangle ++(1.6,.7);
  \node at (\x+.8,.35) {$\lab$};
 }
 \foreach \l/\r/\lab in {0/3.2/U_1,3.2/8/U_2,8/11.2/U_3}{
  \draw[blue!70!black] (\l,.9)--(\l,1.12)--(\r,1.12)--(\r,.9);
  \node[above] at ({(\l+\r)/2},1.12) {$\lab$};
 }
 \foreach \l/\r/\lab in {1.6/4.8/{C_1A_2},6.4/9.6/{C_2A_3}}{
  \draw[orange!70!black] (\l,-.12)--(\l,-.36)--(\r,-.36)--(\r,-.12);
  \node[below] at ({(\l+\r)/2},-.36) {$\lab$};
 }
 \node[above,font=\footnotesize] at (1.6,1.62)
  {$\mathcal T_1:\ B_1\mapsto I$};
 \node[above,font=\footnotesize] at (5.6,1.62)
  {$\mathcal T_2:\ B_2\mapsto I$};
 \node[above,font=\footnotesize] at (9.6,1.62)
  {$\mathcal T_3:\ B_3\mapsto I$};
 \node[anchor=north west,font=\footnotesize] at (0,-.85) {site $1$};
 \node[anchor=north east,font=\footnotesize] at (11.2,-.85) {site $n$};
\end{tikzpicture}
\caption{The partition for $m=3$, with $A_1=C_3=\emptyset$.
Each $\mathcal T_i$ acts inside $U_i=A_iB_iC_i$ and makes
$B_i$ an identity factor. The intervals $C_iA_{i+1}$ between
consecutive $B_i$ are handled by the second decomposition.}
\label{fig:oned-interval-geometry}
\end{figure}

The relaxation in Condition 1 is that we will only require the projection to impose an identity factor on $B_i$, while allowing nontrivial action on the buffers $A_i,C_i$.  Specifically, there are two superoperators we care about:
\begin{itemize}
\item Exact conditional expectation $\mathbb E_{B_i}$:
\[
\mathbb E_{B_i}(X) = \operatorname{Proj}^{\rho_\beta}_{I_{B_i}\otimes\mathcal B(\mathcal H_{B_i^c})}(X)=\text{argmin}_{Y\in I_{B_i}\otimes\mathcal B(\mathcal H_{B_i^c})}\norm{X-Y}_{\rho_\beta}.
\]
This is KMS orthogonal and satisfies the local-gap estimate (Condition 3)       according to~\cite[Lemma III.9, Propositions V.1-V.2]{BC2025}. However, it can act nontrivially outside $B_i$, and projections for disjoint intervals need not commute, meaning Conditions 1 and 2 are not satisfied.

\item Localized commuting projections $\mathcal T_i$: Define
\begin{equation}\label{eq:sketch-local-projection}
 \mathcal T_i:=\mathbb E_{B_i}^{(U_i)}
 :=\left(\operatorname{Proj}^{\rho_{\beta,U_i}^{\mathrm{loc}}}_{
    I_{B_i}\otimes\mathcal B(\mathcal H_{U_i\setminus B_i})}\right)
    \otimes\mathrm{id}_{U_i^c}.
\end{equation}
For a simple tensor, this means
\[
 \mathcal T_i(X_{U_i}\otimes X_{U_i^c})
 =\left(\underset{
 Y\in I_{B_i}\otimes\mathcal B(\mathcal H_{U_i\setminus B_i})}
 {\operatorname{argmin}}
 \norm{X_{U_i}-Y}_{\rho_{\beta,U_i}^{\mathrm{loc}}}\right)
 \otimes X_{U_i^c}.
\]
Here
\[
 \rho_{\beta,U}^{\mathrm{loc}}
 =\frac{e^{-\beta H_U}}{\Tr e^{-\beta H_U}},\qquad
 H_U=\sum_{\{i,i+1\}\subseteq U}h_{i,i+1},
\]
is the Gibbs state on $\mathcal H_U$ of the Hamiltonian restricted
to $U$; it is generally different from the reduced state
$\Tr_{U^c}\rho_\beta$. Because the projection is defined with respect to the local Gibbs state, it acts as the identity outside $U_i$. In particular, $\mathcal T_i$ and $\mathcal T_j$ commute for disjoint intervals $U_i\cap U_j=\emptyset$, so Condition 2 is rigorously satisfied. However, because this projection is defined with respect to the local Gibbs state, it is not KMS orthogonal for the full Gibbs state $\rho_\beta$. This means that Condition 3 is not satisfied exactly.
\end{itemize}

\reviewedit{Both superoperators map any operator to one that is the identity on $B_i$.} In the first step, we won't try to extend the identity factor to the buffers $A_i,C_i$. Instead, we will use the buffers to control the error between $\mathbb E_{B_i}$ and $\mathcal T_i$. Thanks to the buffers $A_i,C_i$ and the locality of the Gibbs state, we can show that $\mathcal T_i$ is close to $\mathbb E_{B_i}$ in the KMS norm. Specifically, for $|A_i|,|B_i|,|C_i|=\Theta(r)=\Theta(\log(2n))$, we have
\begin{equation}\label{eqn:buffer-error}
 \norm{\mathcal T_i-\mathbb E_{B_i}}_{\rho_\beta\to\rho_\beta}
 \leq C\exp\{C|B_i|-cr\log(r+2)\}=\mathcal{O}(n^{-1})
\end{equation}
see Lemma~\ref{lem:oned-local}. This inspires us to choose the length of the buffers and the intervals to be logarithmic in $n$, meaning $\ell=\reviewedit{\Theta(\log(2n))}$.

Now, we can first replace $\mathcal P_i$ by $\mathcal T_i$ in the product expansion \reviewedit{to define} $X_P=\mathcal R_PX$, where
\[
 \mathcal R_P:=\prod_{i\in P}(I-\mathcal T_i)
               \prod_{i\notin P}\mathcal T_i.
\]
This formula is very similar to the idealized decomposition in~\eqref{eqn:decomposition}. Because $\mathcal{T}_i$ are commuting, and approximately satisfy the KMS orthogonality and \reviewedit{local-gap estimate} thanks to~\eqref{eqn:buffer-error}, using \reviewedit{a similar calculation} as in~\eqref{eqn:decomposition-norm}, we can show that
\[
 \begin{aligned}
 \sum_P|P|\norm{X_P}_{\rho_\beta}^2\leq C\sum_i\norm{(I-\mathcal T_i)X}_{\rho_\beta}^2\leq C\bigl[\mathcal D_{\rho_\beta}(X)+\norm X_{\rho_\beta}^2\bigr]
 \end{aligned}
\]
This is stated in Proposition~\ref{prop:oned-first-decomposition}. The additional norm term records the error in replacing the full-state
KMS projection by a local map. The logarithmic buffers make this
error small despite summing over all intervals. This completes the first step of the decomposition.

In summary, in the first step, we successfully decompose $X$ into components $X_P$ that approximately satisfy~\eqref{eqn:decomposition-norm}. However, $X_P$ is not guaranteed to be supported on the union of the intervals $U_i$ with $i\in P$. Instead, it may \reviewedit{act nontrivially} on the buffers $A_i,C_i$ of the intervals $U_i$ with $i\notin P$. This gives a \reviewedit{relaxation} of Condition 1. {The second step will further decompose $X_P$ on these intervening regions, yielding components with controlled support.}

\paragraph{Step 2: Removing the intervening support.}

The first decomposition leaves a specific problem:
\[
 \operatorname{supp}(X_P)
 \subseteq\bigcup_{i\in P}B_i\cup \bigcup_{j=1}^{m-1}D_j,\quad D_j:=C_j\cup A_{j+1}
\]
\reviewedit{Although $P$ may be small, $X_P$ may depend on every interval $D_j$, even when $j,j+1\notin P$.}

Fortunately, if $j,j+1\notin P$, then $X_P$ is already the identity on the main parts $B_j$ and $B_{j+1}$, meaning that
\begin{equation}\label{eqn:XP_structure}
X_P= I_{B_j \cup B_{j+1}} \otimes Y, \quad Y\in\mathcal B\!\left( \mathcal H_{[n]\setminus(B_j\cup B_{j+1})} \right).
\end{equation}
This special structure allows us to remove the dependence on $D_j$ without affecting the identity factors on $B_j$ and $B_{j+1}$, ensure commutativity, and preserve the local-gap estimate simultaneously. In \reviewedit{contrast} to the first step, we will now see $B_j$ and $B_{j+1}$ as buffers for the interval $D_j$ to control the error in removing the dependence on $D_j$ and preserving the local-gap estimate. Specifically, we set
\begin{equation}\label{eqn:G_def}
 \begin{aligned}
 J_P&:=\{\reviewedit{j\in[m-1]}:j,j+1\notin P\},\qquad
 \Lambda_j:=B_jD_jB_{j+1},\\
 \mathcal G_j(X)&:=\mathbb{E}^{(\Lambda_j)}_{\Lambda_j}(X)=I_{\Lambda_j}\otimes
       \Tr_{\Lambda_j}(\rho_{\beta,\Lambda_j}^{\rm loc}X).
 \end{aligned}
\end{equation}
This decomposition can be seen in Figure~\ref{fig:oned-second-decomposition}.

\begin{figure}[!htbp]
\centering
\begin{tikzpicture}[
 every node/.style={font=\small},
 side/.style={draw=orange!65!black,fill=orange!12,line width=.6pt},
 middle/.style={draw=blue!65!black,fill=blue!9,line width=.6pt}]
 \foreach \x/\lab in {0/B_j,4.8/B_{j+1},9.6/B_{j+2}}{
  \draw[middle] (\x,0) rectangle ++(1.6,.85);
  \node at (\x+.8,.59) {$\lab$};
  \node[font=\footnotesize] at (\x+.8,.23) {$I$};
 }
 \foreach \x/\lab in {1.6/C_j,3.2/A_{j+1},6.4/C_{j+1},8/A_{j+2}}{
  \draw[side] (\x,0) rectangle ++(1.6,.85);
  \node at (\x+.8,.425) {$\lab$};
 }
 \draw[blue!70!black] (0,1.03)--(0,1.25)--(6.4,1.25)--(6.4,1.03);
 \node[above] at (3.2,1.25) {$\Lambda_j$};
 \draw[blue!70!black] (4.8,-1.28)--(4.8,-1.5)--(11.2,-1.5)--(11.2,-1.28);
 \node[below] at (8,-1.5) {$\Lambda_{j+1}$};
 \foreach \l/\r/\lab in {1.6/4.8/D_j,6.4/9.6/D_{j+1}}{
  \draw[orange!70!black] (\l,-.12)--(\l,-.36)--(\r,-.36)--(\r,-.12);
  \node[below] at ({(\l+\r)/2},-.36) {$\lab$};
 }
 \node[font=\footnotesize] at (3.2,-.96)
  {$\mathcal G_j:\ D_j\mapsto I$};
 \node[font=\footnotesize] at (8,-.96)
  {$\mathcal G_{j+1}:\ D_{j+1}\mapsto I$};
\end{tikzpicture}
\caption{The second decomposition for $j,j+1\in J_P$.
The blue intervals are already identity factors of $X_P$.
Each $\mathcal G_j$ averages over $\Lambda_j$ with
$\rho_{\beta,\Lambda_j}^{\mathrm{loc}}$, making $D_j$ an identity
factor as well; the factor $I-\mathcal G_j$ allows dependence on $D_j$.
The two averaging regions overlap only on the identity factor $B_{j+1}$.}
\label{fig:oned-second-decomposition}
\end{figure}

If we naively apply $\mathcal G_j$ to general $X$, \reviewedit{it does not give the required estimates} because both conditions 2 and 3 are violated with $O(1)$ error. The key observation is that we only need to apply $\mathcal G_j$ to $X_P$. \reviewedit{Let $W_P$ denote the subspace of operators that are the identity on every $B_h$ with $h\notin P$.} Specifically, we can show:
\[
[\mathcal{G}_i,\mathcal{G}_j](X)=0,\quad i\neq j,\quad \left\|(\mathcal{G}_j-\mathbb{E}_{R_j})X\right\|_{\rho_\beta}=\mathcal{O}(\left\|X\right\|_{\rho_\beta}/n),\quad \text{\reviewedit{$X\in W_P$ and $i,j\in J_P$}}\,,
\]
where $R_j=(B_j)_r\cup D_j\cup(B_{j+1})_\ell$ is the union of the right half of $B_j$, the interval $D_j$, and the left half of $B_{j+1}$. Here, because we see $B_j$ and $B_{j+1}$ as buffers, we need a lower bound on the size of $B_j$ and $B_{j+1}$, which explains the choice of $|B_i|=\Theta(\log(2n))$.

After defining $\mathcal{G}$, we can naturally decompose $X_P$ into components that are supported on different subsets of intervals $D_j$ with $j\in J_P$, similar to~\eqref{eqn:decomposition}:
\[
 Z_{P,T}:=\prod_{j\in T}(I-\mathcal G_j)
          \prod_{j\in J_P\setminus T}\mathcal G_jX_P,
 \qquad
 \sum_{T\subseteq J_P}|T|\norm{Z_{P,T}}_{\rho_\beta}^2
 \leq C\bigl[\mathcal D_{\rho_\beta}(X_P)+ \norm{X_P}_{\rho_\beta}^2\bigr]
\]
The bound is stated in Proposition~\ref{prop:oned-second-decomposition}. Each $Z_{P,T}$ is supported on at most $3|P|+2|T|$ of the original
intervals $U_i$. \reviewedit{This gives the final $\widehat X_S$ in Proposition~\ref{prop:oned-decomposition}} by grouping the $Z_{P,T}$ by their support intervals. {In the next step, we combine the two decompositions, using the estimates from both steps to establish the desired bound for the final decomposition.}

\paragraph{Step 3: Controlling the sum of the dissipation.}

According to the second step, with $\ell=\reviewedit{\Theta(\log(2n))}$, we have a decomposition that guarantees a similar estimation as in~\eqref{eqn:decomposition-norm}. However, the right-hand side of the estimation contains $\reviewedit{\sum_P\bigl[\mathcal D_{\rho_\beta}(X_P)+\left\|X_P\right\|_{\rho_\beta}^2\bigr]}$, which is not directly upper bounded by $\mathcal D_{\rho_\beta}(X)+\norm X_{\rho_\beta}^2$.

Thanks to the approximate orthogonality of the decomposition, which involves an error of $\mathcal{O}(1/n)$, we can show that
\[
 \sum_P\mathcal D_{\rho_\beta}(X_P)
 \leq C\ell\bigl[\mathcal D_{\rho_\beta}(X)+\frac{1}{n}\norm X_{\rho_\beta}^2\bigr]\leq C\log(n)\bigl[\mathcal D_{\rho_\beta}(X)+\norm X_{\rho_\beta}^2\bigr]\,.
\]
This gives an additional factor of $\log(n)$ in the final estimation. Grouping the $Z_{P,T}$ by their
support intervals, and using approximate orthogonality to sum any
selected family, yields Proposition~\ref{prop:oned-decomposition}:
\[
 X=\sum_S\widehat X_S,\
 \operatorname{supp}(\widehat X_S)\subseteq\bigcup_{i\in S}U_i,
\
 \norm{\sum_{|S|>k}\widehat X_S}_{\rho_\beta}^2
 \leq\frac{C\ell}{k+1}
       \bigl[\mathcal D_{\rho_\beta}(X)+\norm X_{\rho_\beta}^2\bigr]\leq \frac{C\ell^2}{(k+1)\ell}\bigl[\mathcal D_{\rho_\beta}(X)+\norm X_{\rho_\beta}^2\bigr]\,,
\]
where $(k+1)\ell$ is \reviewedit{an upper bound on} the number of sites in the union of $k+1$ intervals. Finally, the uniform gap absorbs $\norm X_{\rho_\beta}^2$
into $\mathcal D_{\rho_\beta}(X)$ and gives the desired weight-dissipation bound in Lemma~\ref{lem:framework-weight-dissipation}.

\subsection{Proof of Lemma~\ref{lem:framework-low-weight-norm}: Controlling the coefficients of low-weight operators}

An operator $Y\in V_{\leq s}$ can act on the entire chain, even
though each Pauli string in its expansion is supported on at most
$s$ sites. To prove Lemma~\ref{lem:framework-low-weight-norm},
we therefore bound its Pauli coefficients individually in terms
of its KMS norm. Writing $Y$ in the Pauli basis gives
\[
 Y=\sum_{w(P)\leq s}y_PP,
 \qquad
 \norm Y_\infty\leq\sum_{w(P)\leq s}|y_P|.
\]
There are at most $(Cn)^s$ such strings, so it suffices to bound
each coefficient by $(Cn)^{cs}\norm Y_{\rho_\beta}$. Since Pauli
strings are generally not orthogonal in the KMS inner product,
taking the KMS inner product with $P$ does not directly extract
$y_P$.

Instead, we isolate $y_P$ in two steps. First, Lagrange
interpolation in the Pauli weight operator $\mathcal N$, defined
below, selects the component $Y_w$ of weight exactly $w$. Second,
for a string $P$ of weight $w$, averaging $P^\dagger Y_w$ over
Pauli conjugations on $\operatorname{supp}P$ leaves only $y_PI$.
To control both steps in the KMS norm, we combine the H\"older
inequality in~\cite[Lemma II.2]{BC2025} with the locality of
complex-time evolution in one dimension
from~\cite[Lemma A.2]{BC2025}. Together, these results bound the
KMS norm of multiplication by a single-site Pauli operator
uniformly in the system size.

These two steps have a natural interpretation in quantum Fourier
analysis. Pauli weight plays the role of Fourier degree, and
$e^{-t\mathcal N}$ is the product depolarizing semigroup, which
acts diagonally on the Pauli basis. The corresponding noise
superoperator is defined and identified with a product
depolarizing channel
in~\cite[Definition 44 and Proposition 45]{MO2010}.
Likewise, multiplication by a Pauli string followed by a
normalized partial trace extracts an operator-valued Fourier
coefficient, as described
in~\cite[Definition 27 and Lemma 28]{MO2010}. The corresponding
conditional expectation, with the identity restored on the
traced sites, is precisely the local Pauli average used in our
second step~\cite[Proof of Proposition 6.9]{RWZ2024}.

\paragraph*{Step 1: selecting one weight by a polynomial.}
Write $Y=\sum_{w=0}^sY_w$, where $Y_w$ contains the Pauli strings
of weight exactly $w$. To extract $Y_w$, we first construct a map that multiplies each Pauli string by its weight and then apply a polynomial of that map that vanishes on all other weights. Since the largest weight of a Pauli string in $Y$ is $s$, the polynomial has degree at most $s$, which ensures the norm amplification can be controlled by a constant to the power $s$. The key then is to find the system size dependence of that constant.

In Section~\ref{sec:oned-norm}, the map and the polynomial are defined as:
\[
 \mathcal N(X):=\sum_{i=1}^n
       \bigl(X-\mathbb E_{\omega_i}[\omega_iX\omega_i^\dagger]\bigr),\quad  f_w(x):=\prod_{\substack{0\leq j\leq s\\j\ne w}}
                 \frac{x-j}{w-j},
\]
To control the norm amplification in $\mathcal{N}$, we apply imaginary-time estimates used in~\cite[Lemmas II.2 and A.2]{BC2025} to show that there exists a constant $M_\beta\geq1$ such that $\norm{\mathbb E_{\omega_i}[\omega_iX\omega_i^\dagger]}_{\rho_\beta}\leq \reviewedit{M_\beta^2}\norm X_{\rho_\beta}$ for every $X$.

Combining this and the bound of the absolute coefficients of $f_w$ gives
\[
 \norm{Y_w}_{\rho_\beta}\leq[2n(1+M_\beta^2)]^s\norm Y_{\rho_\beta}\,.
\]
Here $n^s$ comes from the fact that this polynomial has degree $s$.

\paragraph*{Step 2: extracting one coefficient by Pauli averaging.}
Inspired by the well-known partial trace identity:
\[
\frac1{d^2}\sum_{u=0}^{d^2-1}\sigma_{j,u}A\sigma_{j,u}^\dagger
=\frac{I_j}{d}\otimes\operatorname{Tr}_j(A),\quad \forall A\in\mathcal B(\mathcal H).
\]
We extract $y_P$ from $Y_w$ using the identity (Set $A=P^\dagger Y_w$):
\[
 \mathbb E_\omega[\omega P^\dagger Y_w\omega^\dagger]=y_PI.
\]
where $\omega$ \reviewedit{takes} an independent uniformly random single-site
Pauli at each site in $\operatorname{supp}P$\reviewedit{.}

A string $R$ in $Y_w$ survives this average only if it agrees with
$P$ on all $w$ sites of $\operatorname{supp}P$. Since $R$ also has
weight $w$, this forces $R=P$. This explains why selecting the
weight first is useful: strings of larger weight could otherwise
survive outside $\operatorname{supp}P$.~Applying~\cite[Lemmas II.2 and A.2]{BC2025} again, we can control $\left\|\omega P^\dagger Y_w\omega^\dagger\right\|_{\rho_\beta}\leq M_\beta^{3w}\norm{Y_w}_{\rho_\beta}$, where $M_\beta$ is the same constant as in Step 1. This gives the desired bound on $y_P$ in terms of $\norm{Y_w}_{\rho_\beta}$. Summing over all $P$ of weight at most $s$ then gives the final bound on $\norm{Y}_\infty$ in terms of $\norm{Y}_{\rho_\beta}$.

\section{Rapid mixing of one-dimensional local spin systems}\label{sec:proof-of-lemma}

In this section, we prove Lemma~\ref{lem:framework-weight-dissipation} and Lemma~\ref{lem:framework-low-weight-norm}.

\subsection{Proof of Lemma~\ref{lem:framework-weight-dissipation}: Decomposition by intervals}
\label{sec:oned-approximation}
The key step to prove Lemma~\ref{lem:framework-weight-dissipation} is to decompose an arbitrary operator into terms supported on unions of intervals, with a bound on the contribution from many intervals.

\begin{proposition}[Decomposition by intervals]\label{prop:oned-decomposition}
There is a partition $[n]=\bigsqcup_{i=1}^m U_i$ into consecutive
intervals, with $\ell=\max_i|U_i|\leq C\log(2n)$, such that every
operator $X$ admits a decomposition
\begin{equation}\label{eq:oned-decomposition-target}
 X=\sum_{S\subseteq[m]}\widehat X_S,\qquad
 \operatorname{supp}\left(\widehat X_S\right)\subseteq\bigcup_{i\in S}U_i,
\end{equation}
satisfying, for every integer $k\geq0$,
\begin{equation}\label{eq:oned-tail-target}
 \norm{\sum_{|S|>k}\widehat X_S}_{\rho_\beta}^2
 \leq\frac{C\ell}{k+1}[\mathcal D_{\rho_\beta}(X)+\norm X_{\rho_\beta}^2].
\end{equation}
The constant $C$ depends only on $\beta,q$.
\end{proposition}

This proposition implies that
\[
\frac{\mathcal D_{\rho_\beta}(X)+\norm{X}_{\rho_\beta}^2}
{\left\|\hat{X}_S\right\|^2_{\rho_\beta}}
\]
grows linearly with $S$ up to a $\log(n)$ factor. In other words, components with larger support exhibit a larger Dirichlet form and therefore decay at a faster rate. Lemma~\ref{lem:framework-weight-dissipation} then follows directly from Proposition~\ref{prop:oned-decomposition}.

\begin{proof}[Proof of Lemma~\ref{lem:framework-weight-dissipation}]
To bound the distance to $V_{\le s}$, it suffices to construct an
approximating operator in that subspace. First set
$X_0=X-\Tr(\rho_\beta X)I$. Since $I\in V_{\le s}$,
\[
 (I-Q_s)X=(I-Q_s)X_0,\qquad
 \mathcal D_{\rho_\beta}(X_0)=\mathcal D_{\rho_\beta}(X),\qquad
 \norm{X_0}_{\rho_\beta}^2\leq\lambda_{\mathrm{loc}}^{-1}\mathcal D_{\rho_\beta}(X),
\]
where the last bound is the uniform gap in
Lemma~\ref{lem:oned-local-gap} with $R=[n]$.
Apply Proposition~\ref{prop:oned-decomposition} to $X_0$, write its
decomposition as $X_0=\sum_S \widehat X_{0,S}$, and take
\[
 k=\lfloor s/\ell\rfloor,\qquad
 Y=\sum_{|S|\leq k}\widehat X_{0,S}.
\]
Each term of $Y$ is supported on at most $k\ell\leq s$ sites,
so $Y\in V_{\le s}$. Best approximation by $Q_s$ and
\eqref{eq:oned-tail-target} give
\[
 \begin{aligned}
 \norm{X-Q_sX}_{\rho_\beta}^2
 &\leq\norm{X_0-Y}_{\rho_\beta}^2\\
 &\leq\frac{C\ell}{k+1}
       [\mathcal D_{\rho_\beta}(X_0)+\norm{X_0}_{\rho_\beta}^2]\\
 &\leq\frac{C\ell^2(1+\lambda_{\mathrm{loc}}^{-1})}{s+1}
       \mathcal D_{\rho_\beta}(X),
 \end{aligned}
\]
since $\ell(k+1)\geq s+1$. Now $\ell=O(\log(2n))$ concludes the proof.
\end{proof}

In the following part of this section, we prove Proposition~\ref{prop:oned-decomposition}.

\paragraph*{Step 1: Constructing the first decomposition with buffered intervals}

To construct the components in Proposition~\ref{prop:oned-decomposition},
we follow the two-step decomposition strategy described in
Section~\ref{sec:oned-decomposition-sketch}.
Partition $[n]$ into consecutive intervals as in
Figure~\ref{fig:oned-interval-geometry}:
\[
 [n]=\bigsqcup_{i=1}^m U_i,\qquad
 U_i=A_i\sqcup B_i\sqcup C_i,\qquad A_1=C_m=\emptyset,
\]
and write $\ell=\max_i|U_i|$.
We now carry out the first decomposition, on the intervals $B_i$, using
\[
 \mathcal T_i:=\mathbb E_{B_i}^{(U_i)},
\]
with $\mathbb E_{B_i}^{(U_i)}$ defined in~\eqref{eq:sketch-local-projection}.
Recall that:
\begin{itemize}
 \item $\mathcal T_i$ acts only on $U_i$ and projects onto operators
 that are the identity on $B_i$:
 \[
  \ran\mathcal T_i=I_{B_i}\otimes\mathcal B(\mathcal H_{B_i^c}).
 \]
 \item These projections are idempotent and commute, since the
 intervals $U_i$ are disjoint:
 \[
  \mathcal T_i^2=\mathcal T_i,\qquad
  \mathcal T_i\mathcal T_j=\mathcal T_j\mathcal T_i
  \quad(i,j\in[m]).
 \]
\end{itemize}
These maps give the decomposition
\[
 \mathcal R_P:=\prod_{i\in P}(I-\mathcal T_i)
               \prod_{i\notin P}\mathcal T_i,
 \qquad X_P:=\mathcal R_PX
 \quad(P\subseteq[m]).
\]
Commutativity and idempotence imply
\[
 \mathcal R_Q\mathcal R_P=\mathbf1_{\{P=Q\}}\mathcal R_P
 \quad(P,Q\subseteq[m]),
\]
and
\begin{equation}\label{eq:oned-first-decomposition}
 X=\sum_{P\subseteq[m]}X_P,\qquad
 \mathcal T_iX_P=\mathbf1_{\{i\notin P\}}X_P.
\end{equation}
Thus $X_P$ is the identity on every $B_i$ with $i\notin P$, giving
\[
 \operatorname{supp}(X_P)
 \subseteq\bigcup_{i\in P}B_i
       \;\cup\!\bigcup_{j=1}^{m-1}(C_j\cup A_{j+1}).
\]

\reviewedit{We next prove the weighted count and tail estimates for the first decomposition.} This is summarized in the following proposition.

\begin{proposition}[Properties of the first decomposition]
\label{prop:oned-first-decomposition}
The interval lengths can be chosen with $\ell\leq C\log(2n)$ so that,
for every operator $X$, the decomposition
\eqref{eq:oned-first-decomposition} satisfies
\begin{align}
 \sum_P|P|\norm{X_P}_{\rho_\beta}^2
 &\leq C[\mathcal D_{\rho_\beta}(X)+\norm X_{\rho_\beta}^2],
 \label{eq:oned-first-count}\\
 \norm{\sum_{|P|>k}X_P}_{\rho_\beta}^2
 &\leq\frac{C}{k+1}[\mathcal D_{\rho_\beta}(X)+\norm X_{\rho_\beta}^2]
 \qquad(k\geq0).
 \label{eq:oned-first-tail}
\end{align}
Here $k$ is a nonnegative integer and $C$ depends only on $\beta,q$.
The partition is independent of $X$ and $k$.
\end{proposition}

We note that~\reviewedit{\eqref{eq:oned-first-count}} is very similar to that in Proposition~\ref{prop:oned-decomposition}~\eqref{eq:oned-tail-target}. \reviewedit{The key caveat} is that $X_P$ is not guaranteed to be \reviewedit{supported only} on the intervals $U_i$ with $\reviewedit{i\in P}$. This problem will be addressed in the second step of the analysis.

At the heart of the proof are the commutativity and idempotence
of the projections $\mathcal T_i$,
\[
 \mathcal T_i\mathcal T_j=\mathcal T_j\mathcal T_i,
 \qquad \mathcal T_i^2=\mathcal T_i,
\]
and their closeness to the KMS orthogonal projections $\mathbb E_{B_i}$.
We choose the interval lengths so that
\begin{equation}\label{eq:oned-first-projection-comparison}
 \max_i\norm{\mathcal T_i-\mathbb E_{B_i}}_{\rho_\beta\to\rho_\beta}\leq\delta,
 \qquad n\delta\leq1/100.
\end{equation}
The required comparison follows from the locality estimate in~\cite[Lemma III.3]{BC2025}.
\begin{lemma}[Locality of the projections]\label{lem:oned-local}
\cite[Lemma III.3]{BC2025}.
Let $r\geq1$ and suppose $|A_i|\geq r$ for $i>1$
and $|C_i|\geq r$ for $i<m$.
The projections $\mathcal T_i=\mathbb E_{B_i}^{(U_i)}$ satisfy
\[
 \norm{\mathcal T_i-\mathbb E_{B_i}}_{\rho_\beta\to\rho_\beta}
 \leq C\exp\{C|B_i|-cr\log(r+2)\}
 \qquad(i\in[m]),
\]
where $C,c>0$ depend only on $\beta,q$.
\end{lemma}
Take $r=\Theta(\log(2n))$ and $\delta=1/(100n)$.
Choose each nonempty $A_i,C_i$ to have $r+1$ sites and
$8(r+1)\leq|B_i|<18(r+1)$.
For a sufficiently large implicit constant in $r$, depending only
on $\beta,q$, Lemma~\ref{lem:oned-local} gives
\begin{equation}\label{eq:oned-first-locality-applied}
 \max_i\norm{\mathcal T_i-\mathbb E_{B_i}}_{\rho_\beta\to\rho_\beta}
 \leq\delta=\frac1{100n},
 \qquad \ell\leq C\log(2n).
\end{equation}

These properties enter the proof in two steps:
\begin{itemize}
 \item \emph{Approximate orthogonality.}
 Each $\mathbb E_{B_i}$ is a KMS orthogonal projection, and
 \eqref{eq:oned-first-projection-comparison} makes $\mathcal T_i$
 close to it. Together with commutativity and idempotence of the
 $\mathcal T_i$, this implies $\mathcal{R}_P$ behaves like an almost orthogonal decomposition in the KMS norm. This is summarized in the following lemma:
\begin{lemma}[Approximate orthogonality of the decomposition]\label{lem:oned-approx-orth}
For any
 $\mathcal F\subseteq2^{[m]}$ and family
 $Y_P\in\ran\mathcal R_P$, $P\in\mathcal F$,
 \begin{equation}\label{eq:oned-decomposition-norms}
  C^{-1}\sum_{P\in\mathcal F}\norm{Y_P}_{\rho_\beta}^2
  \leq\norm{\sum_{P\in\mathcal F}Y_P}_{\rho_\beta}^2
  \leq C\sum_{P\in\mathcal F}\norm{Y_P}_{\rho_\beta}^2,
 \end{equation}
 where $C$ is a universal constant independent of $m$ and $\mathcal F$.
\end{lemma}

 \item \emph{Bounding the local errors from the local spectral gap.}
 Starting from the weighted count, approximate orthogonality
 \eqref{eq:oned-decomposition-norms} gives
 \[
  \sum_P|P|\norm{X_P}_{\rho_\beta}^2
  =\sum_i\sum_{P\ni i}\norm{X_P}_{\rho_\beta}^2
  \leq C\sum_i\norm{\sum_{P\ni i}X_P}_{\rho_\beta}^2.
 \]
 Using $(I-\mathcal T_i)X=\sum_{P\ni i}X_P$ from
 \eqref{eq:oned-first-decomposition}, we obtain
 \begin{equation}\label{eq:oned-first-count-reduction}
  \sum_P|P|\norm{X_P}_{\rho_\beta}^2
  \leq C\sum_i\norm{(I-\mathcal T_i)X}_{\rho_\beta}^2.
 \end{equation}
 Closeness to $\mathbb E_{B_i}$ in
 \eqref{eq:oned-first-projection-comparison} reduces these local
 errors to $\norm{(I-\mathbb E_{B_i})X}_{\rho_\beta}^2$, with an error
 of order $\delta^2\norm X_{\rho_\beta}^2$ for each $i$.
 We therefore need the local spectral-gap estimate
 \[
  \norm{(I-\mathbb E_{B_i})X}_{\rho_\beta}^2
  \leq\lambda_{\mathrm{loc}}^{-1}\mathcal D_{\rho_\beta,B_i}(X),
 \]
 where $\mathcal D_{\rho_\beta,B_i}(X)$ is the local dissipation defined in \eqref{eqn:local-dissipation}.
 This bound, with $\lambda_{\mathrm{loc}}>0$ independent of $n$
 and valid for every global operator $X$, was proved
 in~\cite[Lemma III.9 and Propositions V.1--V.2]{BC2025}.
 We restate it in the following: For an interval $R\subseteq[n]$, this gap is defined by
\[
 \lambda_R:=
 \inf_{\substack{X\in\mathcal B(\mathcal H_{[n]})\\X\ne\mathbb E_RX}}
 \frac{\mathcal D_{\rho_\beta,R}(X)}{\norm{X-\mathbb E_RX}_{\rho_\beta}^2},\quad \mathcal D_{\rho_\beta,R}(X):=
\sum_{x\in R,u}\norm{[\sigma_{x,u},X]}_{\rho_\beta}^2\,.
\]
Here $X$ can have arbitrary support, and the norm uses the Gibbs
state $\rho_\beta$ of the full chain; only the sites in the commutator
sum $\mathcal D_{\rho_\beta,R}$ are restricted to $R$.
\begin{lemma}[Uniform conditional spectral gap]\label{lem:oned-local-gap}
\cite[Lemma III.9 and Propositions V.1--V.2]{BC2025}.
There is a constant $\lambda_{\mathrm{loc}}>0$, depending only
on $\beta,q$, such that $\lambda_R\geq\lambda_{\mathrm{loc}}$
for every interval $R\subseteq[n]$. Equivalently, every global
operator $X\in\mathcal B(\mathcal H_{[n]})$ satisfies
\begin{equation}\label{eq:oned-local-gap}
 \norm{(I-\mathbb E_R)X}_{\rho_\beta}^2
 \leq\lambda_{\mathrm{loc}}^{-1}\mathcal D_{\rho_\beta,R}(X).
\end{equation}
\end{lemma}
\end{itemize}

We will first use the above two lemmas to prove Proposition~\ref{prop:oned-first-decomposition} and then prove Lemma~\ref{lem:oned-approx-orth}.

\begin{proof}[Proof of Proposition~\ref{prop:oned-first-decomposition}]
To prove the weighted count \eqref{eq:oned-first-count},
apply approximate orthogonality \eqref{eq:oned-decomposition-norms}:
\[
 \sum_P|P|\norm{X_P}_{\rho_\beta}^2
 =\sum_i\sum_{P\ni i}\norm{X_P}_{\rho_\beta}^2
 \leq C\sum_i\norm{\sum_{P\ni i}X_P}_{\rho_\beta}^2.
\]
Equation~\eqref{eq:oned-first-decomposition} gives
$(I-\mathcal T_i)X=\sum_{P\ni i}X_P$, proving
\eqref{eq:oned-first-count-reduction}.
It remains to bound these local errors.
Write
\[
 (I-\mathcal T_i)X=(I-\mathbb E_{B_i})X
                    +(\mathbb E_{B_i}-\mathcal T_i)X.
\]
The comparison \eqref{eq:oned-first-projection-comparison}
and the local spectral gap in Lemma~\ref{lem:oned-local-gap} give
\begin{equation}\label{eq:oned-first-error-bound}
 \begin{aligned}
 \norm{(I-\mathcal T_i)X}_{\rho_\beta}^2
 &\leq2\norm{(I-\mathbb E_{B_i})X}_{\rho_\beta}^2
       +2\delta^2\norm X_{\rho_\beta}^2\\
 &\leq\frac{2}{\lambda_{\mathrm{loc}}}\mathcal D_{\rho_\beta,B_i}(X)
       +2\delta^2\norm X_{\rho_\beta}^2.
 \end{aligned}
\end{equation}

To sum \eqref{eq:oned-first-error-bound}, note that the intervals
$B_i$ are disjoint, so each site occurs at most once in their
commutator sums and
$\sum_i\mathcal D_{\rho_\beta,B_i}(X)\leq\mathcal D_{\rho_\beta}(X)$.
Consequently,
\begin{equation}\label{eq:oned-core-count}
 \sum_i\norm{(I-\mathcal T_i)X}_{\rho_\beta}^2
 \leq\frac{2}{\lambda_{\mathrm{loc}}}\mathcal D_{\rho_\beta}(X)
       +2m\delta^2\norm X_{\rho_\beta}^2
 \leq C[\mathcal D_{\rho_\beta}(X)+n\delta^2\norm X_{\rho_\beta}^2].
\end{equation}
Since $n\delta^2\leq1$, combining \eqref{eq:oned-core-count}
with \eqref{eq:oned-first-count-reduction} proves
\eqref{eq:oned-first-count}.
For the tail bound, apply \eqref{eq:oned-decomposition-norms}
to the terms with $|P|>k$ and then use the weighted count:
\[
 \norm{\sum_{|P|>k}X_P}_{\rho_\beta}^2
 \leq C\sum_{|P|>k}\norm{X_P}_{\rho_\beta}^2
 \leq\frac{C}{k+1}\sum_P|P|\norm{X_P}_{\rho_\beta}^2
 \leq\frac{C'}{k+1}[\mathcal D_{\rho_\beta}(X)+\norm X_{\rho_\beta}^2].
\]
This proves \eqref{eq:oned-first-tail}.
\end{proof}

\begin{proof}[Proof of Lemma~\ref{lem:oned-approx-orth}]
To prove \eqref{eq:oned-decomposition-norms}, we apply
\eqref{eq:oned-first-locality-applied} and use the KMS orthogonal decomposition
\[
 \norm X_{\rho_\beta}^2
 =\norm{\mathbb E_{B_i}X}_{\rho_\beta}^2
  +\norm{(I-\mathbb E_{B_i})X}_{\rho_\beta}^2.
\]
Write $\Delta_i=\mathcal T_i-\mathbb E_{B_i}$, with
$\norm{\Delta_i}_{\rho_\beta\to\rho_\beta}\leq\delta$. Then
\[
 \begin{aligned}
 &\norm{\mathcal T_iX}_{\rho_\beta}^2
  +\norm{(I-\mathcal T_i)X}_{\rho_\beta}^2-\norm X_{\rho_\beta}^2\\
 &\qquad=2\operatorname{Re}\ip{(2\mathbb E_{B_i}-I)X}{\Delta_iX}
          +2\norm{\Delta_iX}_{\rho_\beta}^2.
 \end{aligned}
\]
The same orthogonality gives
$\norm{(2\mathbb E_{B_i}-I)X}_{\rho_\beta}=\norm X_{\rho_\beta}$.
Cauchy--Schwarz therefore yields
\begin{equation}\label{eq:oned-first-single-split}
 \begin{aligned}
 \left|\norm{\mathcal T_iX}_{\rho_\beta}^2
       +\norm{(I-\mathcal T_i)X}_{\rho_\beta}^2-\norm X_{\rho_\beta}^2\right|
 &\leq2\norm X_{\rho_\beta}\norm{\Delta_iX}_{\rho_\beta}
       +2\norm{\Delta_iX}_{\rho_\beta}^2\\
 &\leq(2\delta+2\delta^2)\norm X_{\rho_\beta}^2.
 \end{aligned}
\end{equation}
Set $\xi=2\delta+2\delta^2$. For any $\mathcal F\subseteq2^{[m]}$
and any family $Y_P\in\ran\mathcal R_P$, $P\in\mathcal F$,
iterating $m$ times using Lemma~\ref{lem:projection-iteration}
in Appendix~\ref{app:almost-orthogonality} with \eqref{eq:oned-first-single-split} gives
\begin{equation}\label{eq:oned-first-iteration}
 (1-\xi)^m\norm{\sum_{P\in\mathcal F}Y_P}_{\rho_\beta}^2
 \leq\sum_{P\in\mathcal F}\norm{Y_P}_{\rho_\beta}^2
 \leq(1+\xi)^m\norm{\sum_{P\in\mathcal F}Y_P}_{\rho_\beta}^2.
\end{equation}
The factors $(1\pm\xi)^m$ are bounded above and below by positive
universal constants because $m\delta\leq1/100$.
Thus \eqref{eq:oned-first-iteration} gives
\eqref{eq:oned-decomposition-norms}.
\end{proof}

\paragraph*{Step 2: Decomposing each $X_P$ on the remaining intervals.}

Proposition~\ref{prop:oned-first-decomposition} controls the contribution
from many intervals $B_i$, but the first decomposition only guarantees
\[
 \operatorname{supp}(X_P)
 \subseteq\bigcup_{i\in P}B_i
       \;\cup\!\bigcup_{j=1}^{m-1}(C_j\cup A_{j+1}).
\]
Even when $|P|$ is small, $X_P$ may still act on all the intervals
between consecutive $B_i$. To control this remaining support, we
follow the second decomposition described in
Section~\ref{sec:oned-decomposition-sketch}.

Write $D_j:=C_j\cup A_{j+1}$ for these intervening intervals.
The selected $B_i$ and their adjacent $D_j$ lie in $O(|P|)$
intervals $U_i$, so we only need to further decompose the dependence
on the remaining $D_j$, indexed by
\[
 J_P:=\{j\in[m-1]:j,j+1\notin P\}.
\]
For $j\in J_P$, $X_P$ is already the identity on $B_j$ and $B_{j+1}$.
As in~\eqref{eqn:G_def}, we use this property to remove its dependence
on $D_j$ by averaging over
\[
 \Lambda_j:=B_j\cup D_j\cup B_{j+1}
 \quad(j\in[m-1]),
\]
see Figure~\ref{fig:oned-second-decomposition}.
Using the local projection defined in~\eqref{eq:sketch-local-projection}
with both regions equal to $\Lambda_j$, set
\[
 \mathcal G_j(X)
 :=\mathbb E_{\Lambda_j}^{(\Lambda_j)}(X)
 =I_{\Lambda_j}\otimes
   \Tr_{\Lambda_j}(\rho_{\beta,\Lambda_j}^{\mathrm{loc}}X)
 \quad(j\in[m-1]).
\]
We apply these maps on the subspace with the identity factors
imposed by the first decomposition:
\[
 W_P:=\left\{I_{\cup_{i\notin P}B_i}\otimes X:
 X\in\mathcal B\!\left(\mathcal H_{[n]\setminus\cup_{i\notin P}B_i}\right)
 \right\}.
\]
In particular, $X_P\in W_P$. Recall that:
\begin{itemize}
 \item For $j\in J_P$, $\mathcal G_j$ preserves these identity
 factors and makes $D_j$ an identity factor as well:
 \[
  \mathcal G_j(W_P)\subseteq W_P,\qquad
  \operatorname{supp}(\mathcal G_jX)\cap D_j=\emptyset
  \quad(X\in W_P).
 \]
 \item On $W_P$, these maps are commuting projections:
 \[
  \mathcal G_j^2X=\mathcal G_jX,\qquad
  \mathcal G_j\mathcal G_lX=\mathcal G_l\mathcal G_jX
  \quad(X\in W_P,\ j,l\in J_P).
 \]
\end{itemize}
As in the first decomposition, define, for $T\subseteq J_P$,
\[
 \mathcal J_{P,T}:=
 \prod_{j\in T}(I-\mathcal G_j)
 \prod_{j\in J_P\setminus T}\mathcal G_j
 \quad\text{on }W_P.
\]
Commutativity and idempotence give
\[
 \mathcal J_{P,Q}\mathcal J_{P,T}
 =\mathbf1_{\{Q=T\}}\mathcal J_{P,T}
 \quad(Q,T\subseteq J_P),
\]
and the decomposition
\begin{equation}\label{eq:oned-two-decompositions}
 Z_{P,T}:=\mathcal J_{P,T}X_P,\qquad
 X_P=\sum_{T\subseteq J_P}Z_{P,T}.
\end{equation}

Each component satisfies
\begin{equation}\label{eq:oned-second-support}
 \operatorname{supp}(Z_{P,T})
 \subseteq\bigcup_{i\in P}B_i
  \;\cup\!\bigcup_{j\in T\cup([m-1]\setminus J_P)}
                        D_j.
\end{equation}
Thus $Z_{P,T}$ is supported on $O(|P|+|T|)$ intervals $U_i$.

Now, we prove a similar upper bound result for the second decomposition.

\begin{proposition}[Properties of the second decomposition]
\label{prop:oned-second-decomposition}
The interval lengths can be chosen with $\ell\leq C\log(2n)$ so that,
for every operator $X$ and every $P\subseteq[m]$,
\begin{equation}\label{eq:oned-second-count}
 \sum_{T\subseteq J_P}|T|\norm{Z_{P,T}}_{\rho_\beta}^2
 \leq C[\mathcal D_{\rho_\beta}(X_P)+\norm{X_P}_{\rho_\beta}^2],
\end{equation}
and, for every integer $k\geq0$,
\begin{equation}\label{eq:oned-second-tail-target}
 \norm{\sum_{\substack{T\subseteq J_P\\|T|>k}}Z_{P,T}}_{\rho_\beta}^2
 \leq\frac{C}{k+1}
       [\mathcal D_{\rho_\beta}(X_P)+\norm{X_P}_{\rho_\beta}^2].
\end{equation}
The same partition works for all $P$, and $C$ depends only on $\beta,q$.
\end{proposition}

The proof is similar to that of
Proposition~\ref{prop:oned-first-decomposition}.
Here we need to compare $\mathcal G_jX$ and $\mathbb E_{R_j}X$
for $X\in W_P$, with $\mathbb E_{R_j}$ the KMS orthogonal
projection associated with an interval $R_j$.
Divide each $B_j$ into left and right halves $(B_j)_\ell,(B_j)_r$,
and define
\[
 R_j:=(B_j)_r\cup D_j\cup(B_{j+1})_\ell,
 \qquad \mathbb E_{R_j}=\mathbb E_{R_j}^{([n])}.
\]
The other halves of $B_j,B_{j+1}$ separate $R_j$ from the
exterior of $\Lambda_j$. The following estimate from~\cite{BC2025}
gives the required comparison, similar to Lemma~\ref{lem:oned-local}.

\begin{lemma}[Comparison with the KMS projection]
\label{lem:oned-buffer-removal}
\cite[Lemma IV.1, Corollary I.2, Theorem B.2, and Remark IV.1]{BC2025}.
Set $a=\frac12\min\{|B_j|,|B_{j+1}|\}$.
For every $X$ that is the identity on $B_j\cup B_{j+1}$,
with arbitrary dependence outside $\Lambda_j$,
\[
 \norm{(\mathbb E_{R_j}-\mathcal G_j)X}_{\rho_\beta}
 \leq\left(Ce^{-ca}
       +Ce^{C\ell-ca\log(a+2)}\right)\norm X_{\rho_\beta}.
\]
Here $C,c>0$ depend only on $\beta,q$.
\end{lemma}
Recall the partition chosen after Lemma~\ref{lem:oned-local},
with $r=\Theta(\log(2n))$ and
$\delta=1/(100n)$. Since $|B_i|\geq8(r+1)$,
we have $a\geq4(r+1)$, and Lemma~\ref{lem:oned-buffer-removal} gives
\[
 Ce^{-ca}+Ce^{C\ell-ca\log(a+2)}\leq\frac1{100n}
\]
for a sufficiently large implicit constant in $r$, depending
only on $\beta,q$. This preserves the first comparison and gives
\begin{equation}\label{eq:oned-second-projection-comparison}
 \norm{(\mathbb E_{R_j}-\mathcal G_j)X}_{\rho_\beta}
 \leq\delta\norm X_{\rho_\beta}
 \quad(X\in W_P,\ j\in J_P),
 \qquad\delta=\frac1{100n}.
\end{equation}
The choice is uniform in $P$.

The comparison \eqref{eq:oned-second-projection-comparison}
gives the almost orthogonality stated below; together with the local
spectral gap in Lemma~\ref{lem:oned-local-gap}, it yields
Proposition~\ref{prop:oned-second-decomposition}.

\begin{lemma}[Almost orthogonality of the second decomposition]
\label{lem:oned-second-almost-orthogonality}
Under \eqref{eq:oned-second-projection-comparison}, for every
$P\subseteq[m]$, every $\mathcal F\subseteq2^{J_P}$, and every family
$Y_T\in\ran\mathcal J_{P,T}$, $T\in\mathcal F$,
\begin{equation}\label{eq:oned-second-sum}
 \frac12\sum_{T\in\mathcal F}\norm{Y_T}_{\rho_\beta}^2
 \leq\norm{\sum_{T\in\mathcal F}Y_T}_{\rho_\beta}^2
 \leq2\sum_{T\in\mathcal F}\norm{Y_T}_{\rho_\beta}^2.
\end{equation}
The constants are independent of $P$ and $\mathcal F$.
\end{lemma}

Although the statement of the above lemma is similar to Lemma~\ref{lem:oned-approx-orth}, we still provide the proof for completeness.

\begin{proof}[Proof of Lemma~\ref{lem:oned-second-almost-orthogonality}]
Fix $P\subseteq[m]$. We apply Lemma~\ref{lem:projection-iteration}
to the commuting idempotents $\mathcal G_j|_{W_P}$, $j\in J_P$,
on $W_P$ equipped with the KMS inner product. It remains to verify
\eqref{eq:projection-single-split} for each map on this space.

Write $\Delta_j=\mathcal G_j-\mathbb E_{R_j}$.
For $X\in W_P$, \eqref{eq:oned-second-projection-comparison} gives
$\norm{\Delta_jX}_{\rho_\beta}\leq\delta\norm X_{\rho_\beta}$.
The KMS orthogonal decomposition
$X=\mathbb E_{R_j}X+(I-\mathbb E_{R_j})X$ gives
$\norm{(2\mathbb E_{R_j}-I)X}_{\rho_\beta}=\norm X_{\rho_\beta}$, and hence
\[
 \begin{aligned}
 &\left|\norm{\mathcal G_jX}_{\rho_\beta}^2
       +\norm{(I-\mathcal G_j)X}_{\rho_\beta}^2-\norm X_{\rho_\beta}^2\right|\\
 &\quad=\left|2\operatorname{Re}\ip{(2\mathbb E_{R_j}-I)X}{\Delta_jX}
                  +2\norm{\Delta_jX}_{\rho_\beta}^2\right|\\
 &\quad\leq(2\delta+2\delta^2)\norm X_{\rho_\beta}^2.
 \end{aligned}
\]
Set $\xi=2\delta+2\delta^2$. Thus
\begin{equation}\label{eq:oned-second-single-split}
 (1-\xi)\norm X_{\rho_\beta}^2
 \leq\norm{\mathcal G_jX}_{\rho_\beta}^2+\norm{(I-\mathcal G_j)X}_{\rho_\beta}^2
 \leq(1+\xi)\norm X_{\rho_\beta}^2.
\end{equation}
Let $s=|J_P|$.
For any $\mathcal F\subseteq2^{J_P}$ and any family
$Y_T\in\ran\mathcal J_{P,T}$, $T\in\mathcal F$,
Lemma~\ref{lem:projection-iteration} applied with
\eqref{eq:oned-second-single-split} gives
\begin{equation}\label{eq:oned-second-iteration}
 (1-\xi)^s\norm{\sum_{T\in\mathcal F}Y_T}_{\rho_\beta}^2
 \leq\sum_{T\in\mathcal F}\norm{Y_T}_{\rho_\beta}^2
 \leq(1+\xi)^s\norm{\sum_{T\in\mathcal F}Y_T}_{\rho_\beta}^2.
\end{equation}
As $s\leq n$ and $\delta=1/(100n)$,
$(1-\xi)^s\geq1/2$ and $(1+\xi)^s\leq2$.
Thus \eqref{eq:oned-second-iteration} implies \eqref{eq:oned-second-sum}.
Taking $\mathcal F=2^{J_P}$ and $Y_T=\mathcal J_{P,T}X$ also gives
\[
 \frac12\norm X_{\rho_\beta}^2
 \leq\sum_{T\subseteq J_P}\norm{\mathcal J_{P,T}X}_{\rho_\beta}^2
 \leq2\norm X_{\rho_\beta}^2
 \qquad(X\in W_P).
\]
\end{proof}

Now, we are ready to prove the upper bound for the second decomposition.

\begin{proof}[Proof of Proposition~\ref{prop:oned-second-decomposition}]
For the weighted count \eqref{eq:oned-second-count}, each component
$Z_{P,T}$ occurs in exactly $|T|$ local errors, since
\[
 (I-\mathcal G_j)X_P=\sum_{T\ni j}Z_{P,T}.
\]
The lower bound in Lemma~\ref{lem:oned-second-almost-orthogonality}
therefore gives
\[
 \sum_T|T|\norm{Z_{P,T}}_{\rho_\beta}^2
 =\sum_{j\in J_P}\sum_{T\ni j}\norm{Z_{P,T}}_{\rho_\beta}^2
 \leq2\sum_{j\in J_P}\norm{(I-\mathcal G_j)X_P}_{\rho_\beta}^2.
\]
To bound these local errors, use
\[
 (I-\mathcal G_j)X_P
 =(I-\mathbb E_{R_j})X_P+(\mathbb E_{R_j}-\mathcal G_j)X_P.
\]
Lemma~\ref{lem:oned-local-gap} and
\eqref{eq:oned-second-projection-comparison} imply
\[
 \norm{(I-\mathcal G_j)X_P}_{\rho_\beta}^2
 \leq\frac{2}{\lambda_{\mathrm{loc}}}\mathcal D_{\rho_\beta,R_j}(X_P)
       +2\delta^2\norm{X_P}_{\rho_\beta}^2.
\]
Consecutive $R_j$ use opposite halves of $B_{j+1}$, so the
$R_j$ are disjoint. Summing therefore gives
\[
 \sum_T|T|\norm{Z_{P,T}}_{\rho_\beta}^2
 \leq\frac4{\lambda_{\mathrm{loc}}}\mathcal D_{\rho_\beta}(X_P)
       +4|J_P|\delta^2\norm{X_P}_{\rho_\beta}^2
 \leq C[\mathcal D_{\rho_\beta}(X_P)+\norm{X_P}_{\rho_\beta}^2].
\]
This proves \eqref{eq:oned-second-count}.
Finally, \eqref{eq:oned-second-sum} and $|T|\geq k+1$ give
\[
 \begin{aligned}
 \norm{\sum_{|T|>k}Z_{P,T}}_{\rho_\beta}^2&\leq2\sum_{|T|>k}\norm{Z_{P,T}}_{\rho_\beta}^2\leq\frac2{k+1}\sum_T|T|\norm{Z_{P,T}}_{\rho_\beta}^2\\
 &\leq\frac{C}{k+1}[\mathcal D_{\rho_\beta}(X_P)+\norm{X_P}_{\rho_\beta}^2],
 \end{aligned}
\]
which proves \eqref{eq:oned-second-tail-target}.
\end{proof}

\paragraph*{Step 3: the final decomposition and estimation.}

We now combine $X=\sum_PX_P$ from
\eqref{eq:oned-first-decomposition} with
$X_P=\sum_{T\subseteq J_P}Z_{P,T}$ from
\eqref{eq:oned-two-decompositions}.
To turn the support bound \eqref{eq:oned-second-support} into a
count of intervals $U_i=A_iB_iC_i$, recall that
$J_P=\{j\in[m-1]:j,j+1\notin P\}$ and define
\begin{equation}\label{eq:oned-interval-union}
 \Gamma(P,T):=
 P\cup\!\bigcup_{j\in T\cup([m-1]\setminus J_P)}\{j,j+1\},
 \qquad |\Gamma(P,T)|\leq3|P|+2|T|.
\end{equation}
The indices from $[m-1]\setminus J_P$ contribute only $P$ and
its immediate neighbors, at most $3|P|$ indices in all; each
$j\in T$ adds at most two. For $S\subseteq[m]$, define
\begin{equation}\label{eq:oned-grouped-components}
 \widehat X_S:=\sum_{\Gamma(P,T)=S}Z_{P,T}.
\end{equation}
Then we have
$$
X=\sum_{S\subseteq[m]}\widehat X_S,\qquad
 \operatorname{supp}(\widehat X_S)\subseteq\bigcup_{i\in S}U_i.
$$
Thus this construction satisfies the decomposition and support condition
\eqref{eq:oned-decomposition-target} in
Proposition~\ref{prop:oned-decomposition}.
We now verify its tail bound \eqref{eq:oned-tail-target}. For every
integer $k\geq0$,
\[
 \sum_{|S|>k}\widehat X_S=\sum_{|\Gamma(P,T)|>k}Z_{P,T}.
\]
It suffices to bound the norm of such a sum.
Again, we \reviewedit{use approximate orthogonality} of this decomposition.

\begin{lemma}[Norm of a sum in the final decomposition]
\label{lem:oned-global-sum}
For the chosen partition, every operator $X$ and every
$\mathcal F\subseteq\{(P,T):P\subseteq[m],\ T\subseteq J_P\}$ satisfy
\begin{equation}\label{eq:oned-sum-bound}
 \norm{\sum_{(P,T)\in\mathcal F}Z_{P,T}}_{\rho_\beta}^2
 \leq C\sum_{(P,T)\in\mathcal F}\norm{Z_{P,T}}_{\rho_\beta}^2.
\end{equation}
\end{lemma}
\begin{proof}[Proof of Lemma~\ref{lem:oned-global-sum}]
To prove \eqref{eq:oned-sum-bound}, we apply almost orthogonality
first to the sum over $P$, and then to the sum over $T$ for each
fixed $P$:
\[
 \begin{aligned}
 \norm{\sum_{(P,T)\in\mathcal F}Z_{P,T}}_{\rho_\beta}^2
 &\overset{\mathrm{(i)}}{\leq} C\sum_{P\subseteq[m]}
       \norm{\sum_{T:\,(P,T)\in\mathcal F}Z_{P,T}}_{\rho_\beta}^2\\
 &\overset{\mathrm{(ii)}}{\leq}2C\sum_{(P,T)\in\mathcal F}\norm{Z_{P,T}}_{\rho_\beta}^2.
 \end{aligned}
\]
Step (ii) follows directly from
Lemma~\ref{lem:oned-second-almost-orthogonality}.
For (i), we show that $Z_{P,T}\in\ran\mathcal R_P$ to apply
\eqref{eq:oned-decomposition-norms}.
Indeed, fix $j\in J_P$. If $i\in P$, then $U_i$ and $\Lambda_j$
are disjoint, so $\mathcal T_i$ and $\mathcal G_j$ commute.
If $i\notin P$, then $\mathcal T_i$ fixes $W_P$, which contains
$\ran\mathcal R_P$ and is preserved by $\mathcal G_j$. Thus
\[
 \mathcal T_i\mathcal G_j\mathcal R_P
 =\begin{cases}
  \mathcal G_j\mathcal T_i\mathcal R_P=0,&i\in P,\\
  \mathcal G_j\mathcal R_P,&i\notin P.
 \end{cases}
\]
The product defining $\mathcal R_P$ gives
$\mathcal R_P\mathcal G_j\mathcal R_P=\mathcal G_j\mathcal R_P$.
Hence every factor of $\mathcal J_{P,T}$ preserves
$\ran\mathcal R_P$, and
$Z_{P,T}=\mathcal J_{P,T}\mathcal R_PX\in\ran\mathcal R_P$, as required.
\end{proof}
Similar to~\eqref{eq:oned-first-count} and~\eqref{eq:oned-second-count} in the first two steps, we have the following lemma:
\begin{lemma}[Count for the final decomposition]
\label{lem:oned-global-count}
For the partition above and every operator $X$,
\begin{equation}\label{eq:oned-weighted-bound}
 \sum_{P,T}|\Gamma(P,T)|\norm{Z_{P,T}}_{\rho_\beta}^2
 \leq C\ell[\mathcal D_{\rho_\beta}(X)+\norm X_{\rho_\beta}^2].
\end{equation}
The constant $C$ depends only on $\beta,q$.
\end{lemma}
\begin{proof}
To prove \eqref{eq:oned-weighted-bound}, use the support count
$|\Gamma(P,T)|\leq3|P|+2|T|$ from
\eqref{eq:oned-interval-union}. Almost orthogonality in
Lemma~\ref{lem:oned-second-almost-orthogonality} and the count in
Proposition~\ref{prop:oned-second-decomposition} give
\begin{equation}\label{eq:oned-total-count-reduction}
 \begin{aligned}
 \sum_{P,T}|\Gamma(P,T)|\norm{Z_{P,T}}_{\rho_\beta}^2
 &\leq3\sum_P|P|\sum_T\norm{Z_{P,T}}_{\rho_\beta}^2
       +2\sum_P\sum_T|T|\norm{Z_{P,T}}_{\rho_\beta}^2\\
 &\leq C\sum_P\left[(|P|+1)\norm{X_P}_{\rho_\beta}^2
                         +\mathcal D_{\rho_\beta}(X_P)\right].
 \end{aligned}
\end{equation}
For the first term of \eqref{eq:oned-total-count-reduction},
Proposition~\ref{prop:oned-first-decomposition}
and almost orthogonality \eqref{eq:oned-decomposition-norms} give
\[
 \sum_P(|P|+1)\norm{X_P}_{\rho_\beta}^2
 \leq C[\mathcal D_{\rho_\beta}(X)+\norm X_{\rho_\beta}^2].
\]
Thus, it suffices to prove
\begin{equation}\label{eq:oned-component-energy}
 \sum_P\mathcal D_{\rho_\beta}(X_P)
 \leq C\ell[\mathcal D_{\rho_\beta}(X)+n\delta^2\norm X_{\rho_\beta}^2].
\end{equation}
The key to proving~\eqref{eq:oned-component-energy} is to use
almost orthogonality to control the sum of squared KMS norms in
the Dirichlet forms.

Recall the definition
\[
 \mathcal D_{\rho_\beta}(X)=\sum_{x=1}^n\sum_u
                    \norm{[\sigma_{x,u},X]}_{\rho_\beta}^2,
\]
where $\sigma_{x,u}$ ranges over the nonidentity Pauli operators on site $x$. We first bound the contribution of a single on-site commutator.
\begin{claim}\label{clm:oned-derivative-bound}
Fix a nonidentity Pauli operator $\sigma_{x,u}$ with $x\in U_i$,
and let $\partial X:=[\sigma_{x,u},X]$.
For every operator $X$, with $X_P=\mathcal R_PX$,
\begin{equation}\label{eq:oned-derivative-bound}
 \sum_P\norm{\partial X_P}_{\rho_\beta}^2
 \leq C\norm{\partial X}_{\rho_\beta}^2
      +C\norm{(I-\mathcal T_i)X}_{\rho_\beta}^2,
\end{equation}
where $C$ depends only on $\beta,q$.
\end{claim}
We prove Claim~\ref{clm:oned-derivative-bound} after completing the proof of Lemma~\ref{lem:oned-global-count}.

Summing the bound in Claim~\ref{clm:oned-derivative-bound}
over $x,u$ and grouping sites by $U_i$, we obtain
\[
 \begin{aligned}
 \sum_P\mathcal D_{\rho_\beta}(X_P)
 &=\sum_i\sum_{x\in U_i}\sum_u\sum_P
       \norm{[\sigma_{x,u},X_P]}_{\rho_\beta}^2\\
 &\overset{\text{by }\eqref{eq:oned-derivative-bound}}{\leq}
 C\sum_i\sum_{x\in U_i}\sum_u
       \left(\norm{[\sigma_{x,u},X]}_{\rho_\beta}^2
             +\norm{(I-\mathcal T_i)X}_{\rho_\beta}^2\right)\\
 &=C\mathcal D_{\rho_\beta}(X)
       +C(4^q-1)\sum_i|U_i|\norm{(I-\mathcal T_i)X}_{\rho_\beta}^2\\
 &\leq C\mathcal D_{\rho_\beta}(X)
       +C'\ell\sum_i\norm{(I-\mathcal T_i)X}_{\rho_\beta}^2.
 \end{aligned}
\]
Here each site has $4^q-1$ nonidentity Pauli operators, and $|U_i|\leq\ell$.
Finally, the local spectral gap in Lemma~\ref{lem:oned-local-gap}
and the comparison \eqref{eq:oned-first-locality-applied}
give the local-error bound used in
the proof of Proposition~\ref{prop:oned-first-decomposition},
\[
 \sum_i\norm{(I-\mathcal T_i)X}_{\rho_\beta}^2
 \leq C[\mathcal D_{\rho_\beta}(X)+n\delta^2\norm X_{\rho_\beta}^2].
\]
This proves \eqref{eq:oned-component-energy}.
Substituting this estimate and the bound on
$\sum_P(|P|+1)\norm{X_P}_{\rho_\beta}^2$ into
\eqref{eq:oned-total-count-reduction}, and using $\ell\geq1$ and
$n\delta^2\leq1$, proves \eqref{eq:oned-weighted-bound}.
\end{proof}

\begin{proof}[Proof of Claim~\ref{clm:oned-derivative-bound}]
To prove~\eqref{eq:oned-derivative-bound}, use $X_P=\mathcal R_PX$ to write
\[
 \partial X_P=\mathcal R_P(\partial X)+[\partial,\mathcal R_P]X.
\]
It therefore suffices to bound the two sums in
\[
 \sum_P\norm{\partial X_P}_{\rho_\beta}^2
 \leq2\sum_P\norm{\mathcal R_P(\partial X)}_{\rho_\beta}^2
     +2\sum_P\norm{[\partial,\mathcal R_P]X}_{\rho_\beta}^2.
\]
For the first sum, almost orthogonality
\eqref{eq:oned-decomposition-norms} gives
\[
 \sum_P\norm{\mathcal R_P(\partial X)}_{\rho_\beta}^2
 \leq C\norm{\sum_P\mathcal R_P(\partial X)}_{\rho_\beta}^2
 =C\norm{\partial X}_{\rho_\beta}^2.
\]
For the second sum, we express each commutator in terms of
$Y:=[\partial,\mathcal T_i]X$.
Since an on-site commutator preserves the space of operators that are
the identity on $B_i$, we have
\[
 (I-\mathcal T_i)\partial\mathcal T_i=0,
 \qquad
 Y=-\mathcal T_i\partial(I-\mathcal T_i)X,
 \qquad \mathcal T_iY=Y.
\]
Also, $\partial$ commutes with $\mathcal T_h$ for $h\ne i$,
because these maps act on disjoint intervals.
For $Q\subseteq[m]\setminus\{i\}$, expanding the products defining
$\mathcal R_Q$ and $\mathcal R_{Q\cup\{i\}}$ therefore gives
\[
 \begin{aligned}
 [\partial,\mathcal R_Q]X
 &=\left(\prod_{h\in Q}(I-\mathcal T_h)
         \prod_{h\in[m]\setminus(Q\cup\{i\})}\mathcal T_h\right)Y
   =\mathcal R_QY,\\
 [\partial,\mathcal R_{Q\cup\{i\}}]X
 &=-\mathcal R_QY.
 \end{aligned}
\]
Since $\sum_{Q\subseteq[m]\setminus\{i\}}\mathcal R_Q=\mathcal T_i$,
almost orthogonality \eqref{eq:oned-decomposition-norms} yields
\[
 \begin{aligned}
 \sum_{P\subseteq[m]}\norm{[\partial,\mathcal R_P]X}_{\rho_\beta}^2
 &=2\sum_{Q\subseteq[m]\setminus\{i\}}\norm{\mathcal R_QY}_{\rho_\beta}^2\\
 &\leq2C\norm{\sum_{Q\subseteq[m]\setminus\{i\}}\mathcal R_QY}_{\rho_\beta}^2
 =2C\norm Y_{\rho_\beta}^2.
 \end{aligned}
\]
Using $\norm{\mathcal T_i}_{\rho_\beta\to\rho_\beta}\leq1+\delta$ and
$\norm\partial_{\rho_\beta\to\rho_\beta}\leq C$ from
\cite[Lemmas II.2 and A.2]{BC2025}, we obtain
\[
 \norm Y_{\rho_\beta}^2
 \leq C\norm{(I-\mathcal T_i)X}_{\rho_\beta}^2.
\]
Combining these bounds proves~\eqref{eq:oned-derivative-bound}.
\end{proof}

Finally, we combine the previous results to complete the proof of Proposition~\ref{prop:oned-decomposition}.

\begin{proof}[Proof of Proposition~\ref{prop:oned-decomposition}]
For the partition chosen above,
Lemmas~\ref{lem:oned-global-sum} and~\ref{lem:oned-global-count}
give the required tail bound \eqref{eq:oned-tail-target}:
\[
 \begin{aligned}
 \norm{\sum_{|S|>k}\widehat X_S}_{\rho_\beta}^2
 &=\norm{\sum_{|\Gamma(P,T)|>k}Z_{P,T}}_{\rho_\beta}^2\\
 &\leq C\sum_{|\Gamma(P,T)|>k}\norm{Z_{P,T}}_{\rho_\beta}^2\\
 &\leq\frac{C}{k+1}\sum_{P,T}|\Gamma(P,T)|\norm{Z_{P,T}}_{\rho_\beta}^2\\
 &\leq\frac{C'\ell}{k+1}[\mathcal D_{\rho_\beta}(X)+\norm X_{\rho_\beta}^2].
 \end{aligned}
\]
The partition chosen after Lemma~\ref{lem:oned-local} has
$\ell=O(\log(2n))$ and satisfies \eqref{eq:oned-first-locality-applied}.
The choice verified after Lemma~\ref{lem:oned-buffer-removal}
also ensures \eqref{eq:oned-second-projection-comparison}
for every $P$, while preserving the first comparison.
Together with \eqref{eq:oned-grouped-components}, this proves
Proposition~\ref{prop:oned-decomposition}.
\end{proof}

\paragraph*{Comparison with \cite{BC2025}.}
Our decomposition follows the same local projection idea as~\cite{BC2025};
the middle column combines their spectral gap with our approximate
orthogonality, while the last column records the additional weighted bounds.
\begin{center}
\renewcommand{\arraystretch}{1.8}
\begin{tabularx}{\linewidth}{@{}>{\raggedright\arraybackslash}p{0.14\linewidth}
 >{\raggedright\arraybackslash}X>{\raggedright\arraybackslash}X@{}}
\toprule
Step & Gap consequences of~\cite{BC2025} & This work \\
\midrule
First decomposition
& $\displaystyle \sum_{P\ne\varnothing}\norm{X_P}_{\rho_\beta}^2
   \leq C\mathcal D_{\rho_\beta}(X)$
& $\displaystyle \sum_P|P|\norm{X_P}_{\rho_\beta}^2
   \leq C\mathcal D_{\rho_\beta}(X)$ \\
Second decomposition
& $\displaystyle \sum_{\substack{T\subseteq J_P\\T\ne\varnothing}}
   \norm{Z_{P,T}}_{\rho_\beta}^2\leq C\mathcal D_{\rho_\beta}(X_P)$
& $\displaystyle \sum_{T\subseteq J_P}|T|\norm{Z_{P,T}}_{\rho_\beta}^2
   \leq C\mathcal D_{\rho_\beta}(X_P)$ \\
Combining the two decompositions
& $\displaystyle \norm{X-\Tr(\rho_\beta X)I}_{\rho_\beta}^2
   \leq C\mathcal D_{\rho_\beta}(X)$
& $\displaystyle \norm{\sum_{|S|>k}\widehat X_S}_{\rho_\beta}^2
   \leq\frac{C\ell}{k+1}\mathcal D_{\rho_\beta}(X)$ \\
\bottomrule
\end{tabularx}
\end{center}

\subsection{Proof of Lemma~\ref{lem:framework-low-weight-norm}}
\label{sec:oned-norm}
In this section, we prove Lemma~\ref{lem:framework-low-weight-norm}.

\begin{proof}[Proof of Lemma~\ref{lem:framework-low-weight-norm}]
Our goal is to bound the operator norm of a low-weight operator $Y$ in terms of its KMS norm. To obtain the low-weight component, we first expand $Y=\sum_{w(P)\leq s}y_PP$.
Since every Pauli string has operator norm one,
\[
 \norm Y_\infty\leq\sum_{w(P)\leq s}|y_P|.
\]
There are at most $[1+(d^2-1)n]^s$ such strings. It therefore
suffices to prove $|y_P|\leq(Cn)^s\norm Y_{\rho_\beta}$ for each coefficient,
after choosing $C$ sufficiently large. We first select the part of
$Y$ with weight exactly $w$ and bound its KMS norm. We then extract
each coefficient of that weight by Pauli averaging.

\paragraph*{Step 1: selecting one weight by a polynomial.}
Write
\[
 Y=\sum_{w=0}^sY_w,\qquad
 Y_w:=\sum_{w(R)=w}y_RR.
\]
To control $\norm{Y_w}_{\rho_\beta}$, we express $Y_w$ as a polynomial
in an operator that multiplies each Pauli string by its weight.
For a uniformly random single-site Pauli $\omega_i$, define
\[
 \mathcal N(X):=\sum_{i=1}^n
       \bigl(X-\mathbb E_{\omega_i}[\omega_iX\omega_i^\dagger]\bigr),
 \qquad \mathcal N(P)=w(P)P.
\]
The imaginary-time bound in~\cite[Lemmas II.2 and A.2]{BC2025}
gives a constant $M_\beta\geq1$ such that, for every single-site Pauli $U$,
\[
 \norm{UX}_{\rho_\beta},\ \norm{XU}_{\rho_\beta}\leq M_\beta\norm X_{\rho_\beta}.
\]
Consequently,
$\norm{\mathcal N}_{\rho_\beta\to\rho_\beta}\leq n(1+M_\beta^2)$.
For a fixed $w\in\{0,\ldots,s\}$, construct the polynomial
\[
 f_w(x):=\prod_{\substack{0\leq j\leq s\\j\ne w}}
          \frac{x-j}{w-j}
       =\sum_{r=0}^s a_rx^r.
\]
Since $f_w(k)=\mathbf 1_{\{k=w\}}$ for $k=0,\ldots,s$ and
$\mathcal N(Y_k)=kY_k$, we have
\[
 f_w(\mathcal N)Y=\sum_{k=0}^sf_w(k)Y_k=Y_w.
\]
The nonzero coefficients of $f_w$ alternate in sign, so
\[
 \sum_{r=0}^s|a_r|=|f_w(-1)|=\binom{s+1}{w+1}\leq2^s.
\]
This coefficient bound and the norm bound on $\mathcal N$ give
\begin{equation}\label{eq:oned-weight-component-norm}
 \begin{aligned}
 \norm{Y_w}_{\rho_\beta}
 &=\norm{f_w(\mathcal N)Y}_{\rho_\beta}\leq\sum_{r=0}^s|a_r|\norm{\mathcal N}_{\rho_\beta\to\rho_\beta}^{\,r}
       \norm Y_{\rho_\beta}\\
 &\leq[2n(1+M_\beta^2)]^s\norm Y_{\rho_\beta}
 \leq(Cn)^s\norm Y_{\rho_\beta}.
 \end{aligned}
\end{equation}

\paragraph*{Step 2: extracting one coefficient by Pauli averaging.}
Fix a Pauli string $P$ with $w(P)=w$. To extract $y_P$ from $Y_w$,
choose $\omega$ by taking an independent uniformly random single-site
Pauli at each site in $\operatorname{supp}P$, and the identity elsewhere.
Then
\[
 \mathbb E_\omega[\omega P^\dagger Y_w\omega^\dagger]=y_PI.
\]
Each of $\omega$, $P^\dagger$, and $\omega^\dagger$ has at most $w$
nonidentity factors. The single-site multiplication bound and
$\norm I_{\rho_\beta}=1$ therefore give
\[
 |y_P|
 \leq\mathbb E_\omega\norm{\omega P^\dagger Y_w\omega^\dagger}_{\rho_\beta}
 \leq M_\beta^{3w}\norm{Y_w}_{\rho_\beta}.
\]
Combining this with~\eqref{eq:oned-weight-component-norm} and $w\leq s$
yields
\[
 |y_P|\leq M_\beta^{3s}[2n(1+M_\beta^2)]^s\norm Y_{\rho_\beta}
 \leq(Cn)^s\norm Y_{\rho_\beta}.
\]
After increasing $C$, summing over the at most $[1+(d^2-1)n]^s$
Pauli strings proves $\norm Y_\infty\leq(Cn)^{2s}\norm Y_{\rho_\beta}$. This completes the proof of Lemma~\ref{lem:framework-low-weight-norm}.
\end{proof}

\section{Discussion}

In this work, we develop a new framework for proving rapid mixing of quantum Markov semigroups. In contrast to existing approaches in the quantum literature, our analysis directly exploits the spectral structure of the generator. The key observation is that the slowest-decaying eigenoperators may contribute only weakly to the deviation of the evolving state from the stationary state, whereas the dominant contribution can lie in spectral subspaces associated with substantially faster decay rates. By capturing this nonuniform distribution of the deviation across different spectral modes, our framework yields convergence bounds that can be significantly stronger than those obtained from the spectral gap alone.

To apply this framework to one-dimensional noncommuting Gibbs samplers, inspired by the comparison argument of~\cite{BC2025}, we develop a new decomposition of the operator space into low-weight and high-weight components. \reviewedit{Our weight-dissipation estimate shows that operators with small Rayleigh quotient can be approximated by low-weight Pauli operators.
State normalization limits the contribution that this space can carry
when the chi-squared divergence is large, forcing stronger dissipation.} The use of the Pauli-string decomposition in this work is motivated by the comparison argument in~\ref{cond:framework-comparison}. The underlying framework, however, is not specific to Pauli strings: it can be applied to other operator bases or decompositions whenever an analogous comparison argument can be established.

Our results rigorously establish the possibility of rapid mixing for one-dimensional noncommuting Gibbs samplers \reviewedit{with KMS detailed balance at every fixed positive temperature}. A natural direction for future work is to extend our framework to broader classes of quantum Markov semigroups. As a first step, it would be interesting to investigate whether analogous decompositions can be constructed in other perturbative regimes, such as high-temperature Gibbs samplers in higher dimensions or weakly interacting systems. Such extensions may reveal a more universal design principle for identifying suitable decompositions across a broad class of quantum Markov dynamics. Another interesting direction is to explore connections between our framework and existing techniques for establishing rapid mixing. In particular, it would be valuable to understand whether the spectral structure of the generator exploited here has intrinsic connections to other functional properties of the semigroup, such as hypercontractivity or logarithmic Sobolev inequalities.

\paragraph*{AI assistance.}
This work originated from an investigation of a different problem aided by interactions with GPT-6. The authors formulated the initial research objective and explored whether an approach that emerged from that investigation could be adapted to establish rapid mixing for one-dimensional noncommuting Gibbs samplers based on the KMS-DBC Lindbladian introduced in~\cite{CKG2023}. Through iterative human-AI discussions, the authors utilized GPT-6 to brainstorm the initial proof ideas and to assist with drafting preliminary versions of certain technical calculations appearing in this work. The authors formulated these ideas into the coherent framework developed in Section~\ref{sec:spectral-criterion}, constructed the logically structured proof, and established the main proof strategies presented in Section~\ref{sec:proof-strategy}. The authors guided the overall research direction, critically examined and independently verified all arguments and calculations, and take full responsibility for the final content of the paper.

\appendix
\section*{Appendices}
\section{Iterating the almost orthogonality}
\label{app:almost-orthogonality}
We give the iteration estimate used in both decompositions of
Section~\ref{sec:oned-approximation}.

\begin{lemma}[Iteration for commuting idempotents]
\label{lem:projection-iteration}
Let $\mathcal P_1,\ldots,\mathcal P_m$ be commuting idempotents
on a Hilbert space. Suppose that, for some $0\leq\xi<1$,
\begin{equation}\label{eq:projection-single-split}
 (1-\xi)\norm X^2
 \leq\norm{\mathcal P_iX}^2+\norm{(I-\mathcal P_i)X}^2
 \leq(1+\xi)\norm X^2
 \qquad\text{for every }X\text{ and }i.
\end{equation}
For $S\subseteq[m]$, set
$\mathcal R_S:=\prod_{i\in S}(I-\mathcal P_i)\prod_{i\notin S}\mathcal P_i$.
Then every $\mathcal F\subseteq2^{[m]}$ and every family
$Y_S\in\ran\mathcal R_S$, $S\in\mathcal F$, satisfy
\begin{equation}\label{eq:projection-iteration}
 (1-\xi)^m\norm{\sum_{S\in\mathcal F}Y_S}^2
 \leq\sum_{S\in\mathcal F}\norm{Y_S}^2
 \leq(1+\xi)^m\norm{\sum_{S\in\mathcal F}Y_S}^2.
\end{equation}
\end{lemma}
\begin{proof}
To compare the two quantities in \eqref{eq:projection-iteration},
group the vectors according to $S\cap[h]$ and define
\[
 S_h:=\sum_{Q\subseteq[h]}
 \norm{\sum_{\substack{S\in\mathcal F\\S\cap[h]=Q}}Y_S}^2,
 \qquad 0\leq h\leq m.
\]
Since $S_0=\norm{\sum_{S\in\mathcal F}Y_S}^2$ and
$S_m=\sum_{S\in\mathcal F}\norm{Y_S}^2$, it suffices to show
\[
 (1-\xi)S_{h-1}\leq S_h\leq(1+\xi)S_{h-1}
 \qquad(1\leq h\leq m).
\]
Commutativity and idempotence give $\mathcal P_hY_S=0$ if $h\in S$
and $\mathcal P_hY_S=Y_S$ otherwise. Thus, for $Q\subseteq[h-1]$,
\[
 \begin{aligned}
 \mathcal P_h\left(\sum_{\substack{S\in\mathcal F\\S\cap[h-1]=Q}}Y_S\right)
 &=\sum_{\substack{S\in\mathcal F\\S\cap[h]=Q}}Y_S,\\
 (I-\mathcal P_h)\left(\sum_{\substack{S\in\mathcal F\\S\cap[h-1]=Q}}Y_S\right)
 &=\sum_{\substack{S\in\mathcal F\\S\cap[h]=Q\cup\{h\}}}Y_S.
 \end{aligned}
\]
Applying \eqref{eq:projection-single-split} to each group and summing
over $Q\subseteq[h-1]$ proves the required recurrence.
Iterating it from $h=1$ to $m$ gives \eqref{eq:projection-iteration}.
\end{proof}

\section{R\'enyi-2 decay and a modified logarithmic Sobolev inequality}
\label{app:renyi-log-sobolev}

The $\chi^2$ decay established in the proof of
Theorem~\ref{thm:main_result} is equivalent to exponential convergence
in sandwiched R\'enyi-2 divergence.
For a state $\eta$ and a full-rank state $\rho$, the sandwiched
R\'enyi-2 divergence~\cite{MLDSFT2013} is
\begin{equation}\label{eq:setup-renyi2}
 \widetilde D_2(\eta\Vert\rho)
 :=\log\Tr\!\left[(\rho^{-1/4}\eta\rho^{-1/4})^2\right]
 =\log\bigl(1+\chi^2(\eta\Vert\rho)\bigr).
\end{equation}
For the Gibbs state $\rho_\beta$, let $\eta_t=e^{t\mathcal L}(\eta_0)$,
$A_t=\rho_\beta^{-1/2}(\eta_t-\rho_\beta)\rho_\beta^{-1/2}$, and
$F(t)=\norm{A_t}_{\rho_\beta}^2=\chi^2(\eta_t\Vert\rho_\beta)$.
KMS detailed balance, the comparison~\eqref{eq:framework-comparison},
and Theorem~\ref{thm:framework-enhanced-dissipation} give
\[
 \begin{aligned}
 \frac{d}{dt}\widetilde D_2(\eta_t\Vert\rho_\beta)
 &=\frac{F'(t)}{1+F(t)}
 =-\frac{2\ip{A_t}{-\mathcal L^\dagger A_t}}{1+F(t)}\\
 &\leq-\frac{2c_{\mathrm{cmp}}c_{\mathrm{enh}}}{\log^3(n)}
          \log(1+F(t)).
 \end{aligned}
\]
Thus, with
$\kappa_n=2c_{\mathrm{cmp}}c_{\mathrm{enh}}/\log^3(n)
=\Omega(\log^{-3}(2n))$, we obtain
\begin{equation}\label{eq:setup-renyi-decay}
 \widetilde D_2(\eta_t\Vert\rho_\beta)
 \leq e^{-\kappa_n t}\widetilde D_2(\eta_0\Vert\rho_\beta),
 \qquad t\geq0.
\end{equation}
By~\eqref{eq:setup-renyi2}, exponentiating~\eqref{eq:setup-renyi-decay}
gives
\[
 F(t)\leq\exp\!\bigl[e^{-\kappa_n t}\log(1+F(0))\bigr]-1.
\]
Conversely, taking the logarithm of $1+F(t)$ recovers
\eqref{eq:setup-renyi-decay}, so the two decay estimates are equivalent.

This decay also yields an MLSI with constant $\Omega(\log^{-4}(2n))$.
We first obtain an $L^2\to L^\infty$ norm bound independent of $n$
after time $O(\log^4(2n))$, then combine it with the spectral gap
through the criterion below.

\paragraph*{An MLSI criterion.}
Let $\mathcal L$ be a KMS detailed-balanced Lindblad generator with
unique stationary state $\rho>0$. The relative entropy is
\[
 D(\eta\Vert\rho)=\Tr\bigl[\eta(\log\eta-\log\rho)\bigr].
\]
We define $\alpha_{\mathrm{MLSI}}$ as the largest constant for which
\[
 2\alpha_{\mathrm{MLSI}}D(\eta\Vert\rho)
 \leq-\Tr\bigl[\mathcal L(\eta)(\log\eta-\log\rho)\bigr]
\]
holds for every full-rank state $\eta$. The right-hand side is the
negative time derivative of relative entropy, so MLSI is equivalent to
\[
 D(e^{t\mathcal L}\eta_0\Vert\rho)
 \leq e^{-2\alpha_{\mathrm{MLSI}}t}D(\eta_0\Vert\rho)
 \qquad\text{for every state }\eta_0\text{ and }t\geq0.
\]
To obtain a lower bound on $\alpha_{\mathrm{MLSI}}$, we use
\[
 \norm{e^{t\mathcal L^\dagger}}_{2\to\infty,\rho}
 :=\sup_{X\ne0}\frac{\norm{e^{t\mathcal L^\dagger}X}_\infty}{\norm X_\rho}.
\]
The following criterion requires a bound on this norm at a single
time, in addition to the spectral gap.

\begin{lemma}[MLSI from an $L^2\to L^\infty$ bound]
\label{lem:mlsi-norm-criterion}
\cite[Theorem~3.3, $q=\infty$]{GW2026}
Let $\mathcal L$ be a KMS detailed-balanced Lindblad generator with
unique stationary state $\rho>0$ and spectral gap $\lambda>0$.
If $\norm{e^{\tau\mathcal L^\dagger}}_{2\to\infty,\rho}\leq M$ for some $\tau>0$ and
$M\geq1$, then
\begin{equation}\label{eq:mlsi-norm-criterion}
 \alpha_{\mathrm{MLSI}}
 \geq\frac{1}{2\tau+2(\log M+1)/\lambda}.
\end{equation}
\end{lemma}

The dependence on both $\tau$ and $M$ is essential: our R\'enyi-2
estimate supplies $M=O(1)$ at a polylogarithmic time.

\begin{proposition}[MLSI for the one-dimensional Gibbs sampler]
\label{prop:oned-mlsi}
Under the hypotheses of Theorem~\ref{thm:main_result}, there is a
constant $c>0$, independent of $n$, such that
\[
 \alpha_{\mathrm{MLSI}}\geq\frac{c}{\log^4(2n)}.
\]
Consequently, for every initial state $\eta_0$ and every $t\geq0$,
\[
 D(e^{t\mathcal L}\eta_0\Vert\rho_\beta)
 \leq e^{-2ct/\log^4(2n)}D(\eta_0\Vert\rho_\beta).
\]
The constant may depend on the fixed model parameters $\beta,q,d$.
\end{proposition}
\begin{proof}
To obtain the MLSI bound, we use the size-independent spectral gap
from Fact~\ref{cond:framework-comparison}:
\[
 \lambda:=\lambda_{\rm gap}(\mathcal L)
 \geq c_{\mathrm{cmp}}c_{\mathrm{gap}}>0.
\]
Lemma~\ref{lem:mlsi-norm-criterion} therefore reduces the proof to finding
$\tau=O(\log^4(2n))$ such that
$\norm{e^{\tau\mathcal L^\dagger}}_{2\to\infty,\rho_\beta}\leq\sqrt2$.
We first reduce this norm bound to $\chi^2$ mixing at fixed accuracy.

For a centered Hermitian operator $X$, meaning
$\Tr(\rho_\beta X)=0$, state--observable duality and KMS
Cauchy--Schwarz give
\begin{equation}\label{eq:mlsi-state-observable-duality}
 \begin{aligned}
 \norm{e^{\tau\mathcal L^\dagger}X}_\infty
 &=\sup_{\eta_0\ \mathrm{state}}
       \left|\Tr\bigl[(e^{\tau\mathcal L}\eta_0-\rho_\beta)X\bigr]\right|\\
 &\leq\sup_{\eta_0\ \mathrm{state}}
       \sqrt{\chi^2(e^{\tau\mathcal L}\eta_0\Vert\rho_\beta)}
       \norm X_{\rho_\beta}.
 \end{aligned}
\end{equation}
Extending~\eqref{eq:mlsi-state-observable-duality} to arbitrary
centered operators costs at most a factor $\sqrt2$; adding the scalar
part then gives the required norm bound provided that
\begin{equation}\label{eq:mlsi-fixed-accuracy}
 \sup_{\eta_0\ \mathrm{state}}
 \chi^2(e^{\tau\mathcal L}\eta_0\Vert\rho_\beta)\leq\frac12.
\end{equation}

The decay estimate~\eqref{eq:setup-renyi-decay}, together with the
uniform initial bound $\widetilde D_2(\eta_0\Vert\rho_\beta)
\leq\log\norm{\rho_\beta^{-1}}_\infty=O(n)$, directly gives
\eqref{eq:mlsi-fixed-accuracy} for $\tau=O(\log^4(2n))$.
Finally, the gap $\lambda=\lambda_{\rm gap}(\mathcal L)=\Omega(1)$
and~\eqref{eq:mlsi-norm-criterion} give
\[
 \alpha_{\mathrm{MLSI}}
 \geq\frac{1}{2\tau+(2+\log2)/\lambda}
 =\Omega(\log^{-4}(2n)).
\]
\end{proof}

\end{document}